\documentclass[12pt]{article}
\usepackage{enumitem}
\usepackage{amsmath}
\usepackage{amsthm}
\usepackage{newtxmath}
\usepackage[noend]{algorithmic} 
\usepackage{algorithm,caption}
\algsetup{indent=2em} 

\usepackage{setspace}
\usepackage{adjustbox}
\usepackage{natbib}
\usepackage{libertine} 
\usepackage[T1]{fontenc}
\usepackage{bm}
\usepackage{fancyhdr}
\usepackage{mathrsfs}
\usepackage[margin=1in]{geometry}
\usepackage{color}
\usepackage{comment}
\usepackage{subcaption}
\usepackage{multirow,array}
\usepackage{verbatim}
\usepackage{appendix}
\usepackage{float}
\usepackage{bbm}
\usepackage{graphicx}
\usepackage[font=sc]{caption}
\usepackage{booktabs}       
\usepackage[flushleft]{threeparttable} 
\usepackage{amsfonts}       
\usepackage{nicefrac}       
\usepackage{microtype}      
\usepackage{adjustbox}
\newtheorem{theorem}{Theorem}

\newtheorem{lemma}{Lemma}
\newtheorem{proposition}{Proposition}
\theoremstyle{definition}
\newtheorem{example}{Example}
\newtheorem{remark}{Remark}
\newtheorem{asu}{Assumption}
\newcounter{subassumption}[asu]
\renewcommand{\thesubassumption}{(\textit{\roman{subassumption}})}
\makeatletter
\renewcommand{\p@subassumption}{\theasu}
\makeatother
\newcommand{\subasu}{
  \refstepcounter{subassumption}%
  \thesubassumption~\ignorespaces}

\usepackage[colorlinks = true, linkcolor = blue, urlcolor  = blue,
            citecolor = blue, anchorcolor = blue]{hyperref}

\DeclareMathOperator*{\argmax}{arg\,max}
\DeclareMathOperator*{\argmin}{arg\,min}
\DeclareMathOperator*{\plim}{plim}

\newcommand{\footremember}[2]{%
    \footnote{#2}
    \newcounter{#1}
    \setcounter{#1}{\value{footnote}}%
}
 
\newcommand\Tstrut{\rule{0pt}{2.6ex}}         
\newcommand\Bstrut{\rule[-0.9ex]{0pt}{0pt}}   
\title{Indirect Inference for Nonlinear Panel Models with Fixed Effects\footnote{Conversations with Chuqing Jin inspired and improved this project. I am indebted to Iv{\'a}n Fern{\'a}ndez-Val, Jean--Jacques Forneron and Hiroaki Kaido for patience, guidance and encouragement. For helpful discussions and suggestions, I thank Aureo de Paula, Karun Adusumilli, Jiaying Gu, Eric Hardy, Dennis Kristensen, Yan Liu, Xun Lu, Pierre Perron, Zhongjun Qu, Pascual Restrepo, Marc Rysman, Xiaoxia Shi, Guang Zhang, Beixi Zhou and participants in numerous seminars, reading groups and job interviews. All errors are mine.}}
\author{Shuowen Chen\footremember{alley}{Boston University. Email: swchen@bu.edu.}}
\begin{document}
\maketitle
\begin{abstract}
    Fixed effect estimators of nonlinear panel data models suffer from the incidental parameter problem. This leads to two undesirable consequences in applied research: (1) point estimates are subject to large biases, and (2) confidence intervals have incorrect coverages. This paper proposes a simulation--based method for bias reduction. The method simulates data using the model with estimated individual effects, and finds values of parameters by equating fixed effect estimates obtained from observed and simulated data. The asymptotic framework provides consistency, bias correction, and asymptotic normality results. An application and simulations to female labor force participation illustrates the finite--sample performance of the method. 

\end{abstract}

\clearpage
\section{Introduction}\label{Section: Introduction}
Panel data refers to data on multiple entities (e.g., individuals, firms, etc.) observed at two or more time periods. Unobserved heterogeneity across entities often accounts for a large fraction of the variation in panel data. When this heterogeneity is correlated with the explanatory variables in the regression specifications, the resulting omitted variable bias renders point estimates inconsistent. 

Adding individual fixed effects, $\alpha_{i0}$'s, is the main approach to control for time--invariant unobserved heterogeneity in panel data models. Compared to other approaches like random effects and correlated random effects, the fixed effect approach does not impose distributional assumptions on $\alpha_{i0}$'s or restrict their relationships with other explanatory variables. Instead, each $\alpha_{i0}$ is treated as a parameter to be estimated. However, because the number of $\alpha_{i0}$'s increases with the sample size and each $\alpha_{i0}$ is estimated using only entity $i$'s time series observations, adding fixed effects introduces the incidental parameter problem in estimating the vector of parameters of interest $\theta_{0}$. It has two consequences for applied research: (1) point estimates are subject to large biases, and (2) confidence intervals have incorrect coverages. 

This paper proposes a new method to debias fixed effect estimators in a class of nonlinear panel models. The method is named \textit{indirect fixed effect estimation} and features two main steps: the first one is to simulate data by using estimated individual effects $\widehat{\alpha}_{i}$'s from the observed data. The second step is to find the vector of parameters that matches the fixed effect estimators using observed and simulated data. 

The method has two advantages: first, it does not require an explicit characterization of the bias term, which can be hard to derive in complex models. Instead, the method finds the solution by automatically correcting the bias because the vector of parameter values that is the closest to $\theta_{0}$ renders similar bias in fixed effect estimations. Second, standard errors can be derived using the delta method, so there is no need to use the bootstrap, which is computationally intensive. 

The two properties are inherited from a precedent simulation--based estimation approach called indirect inference, which was first developed by \cite{GourierouxMonfortRenault1993} and \cite{Smith1993}. In a nutshell, indirect inference uses an auxiliary model to summarize the statistical properties of the observed data and simulated data, and finds values of model parameters that match the parameters of the auxiliary model, estimated using the observed and simulated data, in terms of a minimum--distance criterion function. Because the same regression is run on observed and simulated data, matched estimators have the same bias structure and thus the bias gets cancelled. 

The theory of indirect inference, however, is not directly applicable to nonlinear panel models, which are widely used in various fields of economics like industrial organization and labor. Because the individual effects cannot be differenced out, data simulations seem infeasible without imposing a parametric specification on their distributions, and the bias term is a complicated function of $\theta_{0}$ and $\alpha_{i0}$'s. 

To simulate data, this paper proposes using the estimated individual effects $\widehat{\alpha}_{i}$'s. These are informative proxies for the unknown individual effects $\alpha_{i0}$'s because they become more accurate estimates when each individual's number of time series observations $T$ grows large. Intuitively speaking, although data simulated using $\widehat{\alpha}_{i}$'s do not perfectly mimic the observed data, such a difference vanishes when $T$ increases.

The indirect fixed effect estimator then debiases by matching the fixed effect estimates using observed and simulated data. This brings two advantages for the implementation and theoretical analysis of the new estimator. First, the minimum--distance criterion function for matching is just--identified because the dimensions of the fixed effect estimates are identical. Therefore, there is no need to consider an estimation of an optimal weighting matrix. It further implies that the matching can be made as exact as machine precision permits. The second advantage is with respect to the relationship between the vector of parameters of interest $\theta_{0}$ and the unique maximizer of the limiting log--likelihood function for fixed effect estimation. To back out point estimates of $\theta_{0}$ from fixed effect estimators using simulated data, this relationship should be invertible. Because the unique maximizer is $\theta_{0}$, the relation turns out to be an identity function. Therefore, invertibility is satisfied trivially. 

This paper derives consistency, bias correction and asymptotic normality results for the indirect fixed effect estimator. As usual in the indirect inference literature, consistency requires that the fixed effect estimates using observed and simulated data converge to the unique maximizer of the limiting log likelihood. Although the pointwise convergence of $\widehat{\theta}$ to $\theta_{0}$ is a standard result in the large--$T$ panel literature, three important differences arise in the analysis of fixed effect estimates using simulated data and pose theoretical challenges.

First, the simulated data are generated using $\widehat{\alpha}_{i}$'s instead of $\alpha_{i0}$'s. To justify this practice, the corresponding log likelihood function should uniformly well--approximate the one rendered by data simulated using the true individual effects. Otherwise, simulated fixed effect estimator could not be pointwise convergent. The proof of this statement, however, is complicated by the fact that the log likelihood function using simulated data is typically nonsmooth for important types of nonlinear panel models, with binary choice models as leading examples. Intuitively speaking, when the dependent variable is discrete, a small change in the parameter values can lead to discrete changes in the simulated data. As a result, the sample log likelihood function using simulated data is discontinuous.

Simulations often generate discontinuous objective functions \citep[e.g.,][]{McFadden1989, PakesPollard1989}, but this paper confronts a second difference: the fixed effect estimator using simulated data is nonsmooth with respect to the parameters of the data generating process (DGP). Therefore, standard proof strategies in the panel literature \citep[e.g.,][]{HahnNewey2004, HahnKuersteiner2011} cannot be directly applied to characterize its limiting behavior.  

Empirical process theory provides ample tools to handle nonsmoothness functions and moments in econometrics \citep{Andrews1994}, but the analysis of a nonsmooth fixed effect estimator is further complicated by the third difference: the presence of incidental parameters, whose number increases with the sample size $n$. 

To prove uniform convergence with nonsmoothness, this paper follows \cite{Newey1991} by establishing pointwise convergence and stochastic equicontinuity of the fixed effect estimator in the simulation world. Intuitively speaking, pointwise convergence is equivalent to uniform convergence for any finite number of grid points, but without smoothness, the gap between any two grids can behave rather erratically. The stochastic equicontinuity condition is hence required to restrict such behaviors in probability.

The theoretical analysis of the indirect fixed effect estimator relies on some key structures of the panel data and the log likelihood function. Under the assumption that panel data are independent along the cross section dimension, this paper first justifies data simulation with $\widehat{\alpha}_{i}$'s by proving that the corresponding log likelihood function uniformly approximates the one from simulated data generated by $\alpha_{i0}$'s. As such, a uniform law of large number can be established and pointwise convergence in the simulation world follows from the standard consistency argument \citep{NeweyMcFadden1994}. To verify the stochastic equicontinuity condition of fixed effect estimators using simulated data, this paper uses the concavity property of the profiled log likelihood to verify one of the primitive conditions for stochastic equicontinuity in \cite{Andrews1994}. The proof strategy might be of independent interest. 

Regarding the asymptotic unbiasedness and normality of the new estimator, due to non--smoothness in the simulation world, the conventional strategy in indirect inference that relies on the implicit function theorem \citep[e.g.,][]{GourierouxMonfortRenault1993} is not directly applicable. A regularity conditions is thus imposed, which, combined with consistency, allows to explore bias correction and asymptotic normality through the lens of fixed effect estimates. More specifically, the fixed effect estimators in both worlds have the same structures regarding the bias terms and influence functions. The difference is that the ones using the real data are functions of $\theta_{0}$ and $\alpha_{i0}$'s while those using the simulated data are functions of $\theta_{0}$ and $\widehat{\alpha}_{i}$'s. 

This paper currently imposes two high--level conditions to ensure that the bias term and the influence function from data simulated using $\theta_{0}$ and $\widehat{\alpha}_{i}$'s is uniformly close to their infeasible counterparts from data simulated using $\theta_{0}$ and $\alpha_{i0}$'s with asymptotically negligible approximation errors. The infeasible bias converges to the same probability limit as does the bias obtained from observed data, while the infeasible influence function converges to the same normal distribution as does the influence function from observed data. Therefore, the theory of indirect inference can be invoked to establish bias cancellation and asymptotic normality.

Like other simulation--based estimation methods, the asymptotic variance of the new estimator is inflated by the inverse of the number of simulation draws. The result can also be intepreted as a reflection of the classic bias--variance tradeoff. As shown in the application and Monte Carlo simulations, however, the finite--sample performance of the indirect fixed effect estimator is comparable to the leading methods in terms of bias correction and outperforms half--panel bias correction methods in terms of standard errors. 

\subsection*{Related Literature}
The indirect fixed effect estimator presented in this paper combines four strands of literature, of which this section provides a non--exhaustive overview. The incidental parameter problem is first discussed by \cite{NeymanScott1948}. When $T$ is fixed, fixed effect estimators of nonlinear models are in general inconsistent because estimation errors of $\widehat{\alpha}_{i}$'s do not vanish even when the cross--section sample size $n$ is very large \citep{Chamberlain1984, Lancaster2000}. Only some special models like static linear and logit specifications feature fixed--$T$ consistent estimators \citep{Anderson1970}. A key insight of the large--$T$ panel data literature is that the incidental parameter problem becomes an asymptotic bias problem when $T$ grow with the sample size $n$. When $n$ and $T$ grow at the same rate, fixed effect estimators are consistent and asymptotically normal, but they have a bias comparable to standard errors.

In the search for asymptotically unbiased estimators, there are two leading approaches. For certain types of models, the bias terms have been characterized analytically and corrected using a plug--in approach \citep{HahnKuersteiner2002, HahnNewey2004, Fernandez-Val2009, HahnKuersteiner2011}. However, such terms can be hard to derive for complicated models. Under further sampling and regularity conditions, bias terms can be automatically corrected using jackknife. For example, \cite{HahnNewey2004} proposed leave--one--out panel jackknife for data that do not have dependencies among observations of the same unit. \cite{DhaeneJochmans2015} relaxed the assumption to stationarity along the time series, and proposed a half--panel method. Under an unconditional homogeneity assumption, \cite{Fernandez-ValWeidner2016} allowed for two--way fixed effects and propose a jackknife method that corrects biases from both dimensions. See \cite{ArellanoHahn2007} and \cite{Fernandez-ValWeidner2018} for recent surveys. Standard errors are typically obtained using panel bootstrap, which can be computationally intensive.\footnote{For example, to obtain one debiased point estimate, fixed effect estimations are run three times: one for the whole sample, and twice for the two split samples. If the number of bootstraps is set to be 500, then the total number of fixed effects estimations becomes 1500. In addition, in practice it is often recommended to use multiple sample splits to improve the finite--sample performance.} 

Another popular simulation--based method that can achieve bias correction is bootstrap \citep{Horowitz2001, Horowitz2019}. \cite{GoncalvesKaffo2015} proposed the bootstrap bias correction methods for dynamic linear panel models without covariates. \cite{KimSun2016} proposed a parametric bootstrap bias correction (BBC) method for the nonlinear panel models considered in this paper. Compared to the indirect fixed effect estimator, the BBC estimator mimics the bias term and removes it from the fixed effect estimate explicitly, and thus the proof strategies are very different. 


Second, this paper extends the existing theory and practice of indirect inference. Since the introduction of the method, its asymptotic theory has mainly been focused on times series data \citep{GourierouxMonfortRenault1993, Smith1993, GallantTauchen1996}. Some recent papers explore asymptotic properties in panel data with discrete dependent variables, but there are two key differences with this paper. First, their settings hold time series dimension fixed and study different types of models. For example, \cite{GII2018} did not consider models with fixed effects, \cite{FrazierOkaZhu2019} imposed normality on individual effects, and \cite{TaberSauer2021} assumed a bivarite normal distribution on the types of individuals. Second, they deal with nonsmoothness by smoothing the discontinuous parts and showing that the resultant bias can be corrected. 

\cite{GourierouxPhillipsYu2010} is the first paper that establishes theoretical properties of indirect inference for a class of large--$T$ panel models. They applied indirect inference to dynamic panel linear models, whose fixed effect estimators are known to be biased \citep{Nickell1981}. The linear structure allows them to eliminate the individual fixed effects $\alpha_{i0}$'s by first--difference. As such, $\alpha_{i0}$'s do not show up in the bias term, and data can be simulated without information on them. However, first difference does not remove the $\alpha_{i0}$'s to nonlinear panel models, and this paper fills the gap by extending the theory to handle the presence of $\alpha_{i0}$'s in data simulation and the bias term. 

Indirect inference is popular in various fields of economics, including empirical industrial organization \citep{CollardWexler2013}, labor economics \citep{AltonjiSmithVidangos2013} and macroeconomics \citep{GuvenenSmith2014, BergerVavra2019}. However, finding an informative auxiliary model is not a trivial task, and researchers often have to assume invertibility of the limiting relationship between auxiliary parameters and parameters of interest. This paper provides an alternative choice, namely the log likelihood function from the nonlinear panel model, for researchers that employ panel data with fixed effects. The estimation procedures are simple to implement as fixed effect estimation schemes are available in free software like R and Julia. 

Nonsmooth objective functions are common in econometrics, and empirical process methods are standard tools for asymptotic analysis \citep{Andrews1994, NeweyMcFadden1994, VanDerVaartWellner1996}. The seminal work on simulation--based methods by \cite{PakesPollard1989} is predicated on the independence assumption of cross section data and therefore is not suitable for panel data, which feature dependence for each individual time series. \cite{DedeckerLouhichi2002} provided an overview of maximal inequalities for empirical central limit theorems for dependent data. \cite{KatoGalvaoMontesRojas2012} provided new stochastic inequalities for mixing sequences and also established stochastic equicontinuity in the presence of nuisance parameters, but their analysis focused on a different class of nonlinear models, namely panel quantile regression models.

Simulation--based methods like simulated method of moments \citep{McFadden1989, PakesPollard1989, LeeIngram1991, DuffieSingleton1993} and indirect inference are widely used to estimate models that do not render tractable moments or likelihood functions. See \cite{GourierouxMonfort1997} for an overview. These methods typically require models to be fully specified, but economic theory does not always provide guidance on functional forms, distributions of shocks or measurement error of observed data. Therefore, the resultant estimators can be subject to misspecification. 

This paper considers a class of nonlinear panel models that do not impose distributional assumptions on the individual effects, and hence contributes to a burgeoning literature that considers simulations for models that relax the full parametric specifications in various ways. \cite{DridiRenault2000} and \cite{ DridiGuayRenault2007} embedded the semiparametric structural model into a full model for data simulation, and proposed an encompassing principle where parameters of interest are consistently estimated even though nuisance parameters are inconsistently estimated due to misspecification of the full model. \cite{Schennach2014} considered parameters estimation in moment conditions that contain unobservable variables, and proposed a simulation--based method that constructs equivalent moments involving only observable variables. \cite{GospodinovKomunjerNg2017} considered parameter estimation of autoregressive distributed lag models in which covariates are contaminated by serially correlated measurement errors. They proposed a method such that simulated covariates preserve the dependence structure observed in the data even though the dynamics of latent covariates or measurement errors are not specified. \cite{Forneron2020} approximated the distribution of shocks by sieves and proposes a sieve--SMM estimator that jointly estimates structural parameters and the distribution of shocks. 

\subsection*{Structure of the Paper}
The rest of the paper proceeds as follows: Section~\ref{Section: Model and FE} introduces the model and describes the fixed effect estimator and incidental parameter problem. Section~\ref{Section: Overview} provides an overview of the indirect fixed effect estimator and its implementation. Section~\ref{Section: Asymptotics} presents the theoretical properties of the estimator. Section~\ref{Section: Application} applies the method to a dataset on labor force participation to illustrate the finite--sample properties of the estimator. Section~\ref{Section: Monte Carlo} uses numerical simulations to compare the new estimator with other bias correction methods. Section~\ref{Section: Conclusion} concludes and discusses open questions. Appendices~\ref{Appendix: Auxiliary Lemmas}, \ref{Appendix: Main Theorem Proof} and \ref{Appendix: Computation} consist of proofs and computation details. 
\section{Nonlinear Panel Model and Fixed Effect Estimator} \label{Section: Model and FE}
This section starts with a description of nonlinear panel models with fixed effects. Let the data observations be denoted by $\{z_{it}=(y_{it}, x_{it})$: $i=1,\dots,n; t=1,\dots,T\}$, where $y_{it}$ is the dependent variable and $x_{it}$ is a $p\times 1$ vector of explanatory variables. The observations are independent across entity $i$ and weakly dependent across time $t$. The DGP of outcome $y_{it}$ takes the following form:
\begin{equation} \label{structural model}
    y_{it} \mid x^{T}_{i},\alpha_{i},\theta \sim f(\cdot \mid x_{it}; \alpha_{i}, \theta),
\end{equation}
where $x_{i}^{T}:=(x_{i1},\dots,x_{iT})$, $\theta$ is a $p\times 1$ vector of model parameters and $\alpha_{i}$ is a scalar individual effect. The explanatory variable $x_{it}$ is strictly exogenous. The model is semiparametric in that neither the distribution of $\alpha_{i}$ nor its relationship with $x_{it}$ is specified. The conditional density $f$ denotes the parametric part of the model and its form depends on the parametric family of distributions $\{u_{it}\}$. Depending on the specification of $f$, this type of models have been used to study many different questions of economic interest. 

\begin{example}[Discrete Choice Model]\rm
Let $y_{it}$ denote a binary variable and $F_{u}$ a cumulative distribution function (CDF), e.g., the standard normal or logistic distribution. Suppose the binary variable is generated by the following single index process with additive individual effects:
$$y_{it}=\boldsymbol{1}\{x'_{it}\theta + \alpha_{i} \geq u_{it}\}, \quad u_{it} \mid x^{T}_{i}, \alpha_{i} \sim F_{u},$$
where $\boldsymbol{1}\{\cdot\}$ denotes the indicator function. Then the conditional distribution of $y_{it}$ is expressed as 
$$f(y_{it}\mid x_{it}; \alpha_{i}, \theta) = F_{u}(x'_{it}\theta+\alpha_{i})^{y_{it}}(1-F_{u}(x'_{it}\theta+\alpha_{i}))^{1-y_{it}}.$$
\cite{HelpmanMelitzRubinstein2008} modeled a country's export decision as  and estimate a gravity equation with country fixed effects. \cite{Fernandez-Val2009} used a  specification to estimate the determinants of females' labor force participation decisions in the presence of time--invariant heterogeneity such as willingness to work and ability. \cite{CollardWexler2013} used a binary logit specification with market--fixed effects to study whether a ready--mix concrete plant decides to be active in a market. 
\end{example}

\begin{example}[Poisson Regression Model]\rm
The Poisson distribution is useful in modeling count data. Let $y_{it}$ denote the number of arrivals of new events within a certain time interval for entity $i$ in year $t$. For $\lambda_{it}=\exp(x'_{it}\theta+\alpha_{i})$, the conditional density is modeled as
$$f(y_{it}\mid x_{it}; \alpha_{i}, \theta) = \frac{\lambda^{y_{it}}_{it}\exp(-\lambda_{it})}{y_{it}!}\boldsymbol{1}\{y_{it}\in\{0, 1, \dots\}\}.$$
Using the number of citation--weighted patents as a proxy for innovation, \cite{ABBGH2005} employed this specification to study the relationship between innovation and competition with industry fixed effects. 
\end{example}

Model~(\ref{structural model}) admits a log likelihood function. The true values of the parameters, denoted by $\theta_{0}$ and $\alpha_{i0}$'s, are one solution to the population conditional maximum likelihood problem
\begin{equation} \label{population criterion}
    (\theta_{0}, \alpha_{10},\dots,\alpha_{n0})\in\argmax_{(\theta, \alpha_{1},\dots,\alpha_{n})\in\Theta\times\Gamma^{n}_{\alpha}}\frac{1}{nT}\sum^{n}_{i=1}\sum^{T}_{t=1}\mathbb{E}\Big[\ln f(y_{it}\mid x_{it}; \alpha_{i}, \theta)\Big],
\end{equation}
where $\Theta$ and $\Gamma_{\alpha}$ are the parameter spaces for $\theta$ and $\alpha_{i}$ respectively, the expectation is with respect to the distribution of the observed data, conditional on the unobserved effects and initial conditions. Section~\ref{Section: Asymptotics} discusses assumptions under which the log likelihood function is concave in all parameters and the solution uniquely exists. The indirect fixed effect estimator relies on the uniqueness condition for consistency.
\subsection{The Fixed Effect Estimator}
The fixed effect estimator of $\theta$ is obtained by doing maximum likelihood estimation on the sample analog of the population problem~(\ref{population criterion}), treating each $\alpha_{i}$ as a parameter to be estimated. 
\begin{equation*} 
    (\widehat{\theta},          \widehat{\alpha}_{1},\dots,\widehat{\alpha}_{n})\in\argmax_{(\theta, \alpha_{1},\dots,\alpha_{n})\in\Theta\times\Gamma^{n}_{\alpha}}\frac{1}{nT}\sum^{n}_{i=1}\sum^{T}_{t=1}\ln f(y_{it}\mid x_{it}; \alpha_{i}, \theta).
\end{equation*}
To facilitate theoretical analysis, this equation is rewritten such that the individual effects are profiled out. More specifically, given $\theta$, the optimal $\widehat{\alpha}_{i}(\theta)$ for each $i$ is defined as  
\begin{equation*} 
    \widehat{\alpha}_{i}(\theta)\in\argmax_{\alpha\in\Gamma_{\alpha}}\frac{1}{T}\sum^{T}_{t=1}\ln f(y_{it} \mid x_{it}; \alpha, \theta).
\end{equation*}
The estimators $\widehat{\theta}$ and $\widehat{\alpha}_{i}$ are then
\begin{equation} \label{theta_hat expression} 
    \widehat{\theta}\in\argmax_{\theta\in\Theta}\frac{1}{nT}\sum^{n}_{i=1}\sum^{T}_{t=1}\ln f(y_{it} \mid x_{it}; \widehat{\alpha}_{i}(\theta), \theta), \quad \widehat{\alpha}_{i}=\widehat{\alpha}_{i}(\widehat{\theta}).
\end{equation}
Section~\ref{Section: Asymptotics} discusses assumptions under which these estimators exist and are unique with probability approaching one as $n$ and $T$ become large. 
\subsection{The Incidental Parameter Problem}
In panel models, the individual effects are incidental parameters, i.e., nuisance parameters whose dimension grows with the number of cross sectional observations $n$. As equation~(\ref{theta_hat expression}) shows, the fixed effect estimator $\widehat{\theta}$ cannot generally be separated from the estimator of individual effects $\widehat{\alpha}_{i}$'s. Because each $\widehat{\alpha}_{i}$ is only estimated using the $T$ observations for entity $i$, its estimation error does not vanish if $T$ is fixed, even as $n$ grows. These estimation errors in turn contaminate $\widehat{\theta}$. This is the incidental parameter problem for fixed effects estimation. Mathematically, 
\begin{equation*}
    \widehat{\theta}\xrightarrow{p}\theta_{T}:=\argmax_{
\theta\in\Theta}\plim_{n\rightarrow\infty}\frac{1}{n}\sum^{n}_{i=1}\Big(\frac{1}{T}\sum^{T}_{t=1}\ln f(y_{it}\mid x_{it}, \theta, \widehat{\alpha}_{i}(\theta))\Big),
\end{equation*}
whereas
$$\theta_{0}:=\argmax_{
\theta\in\Theta}\plim_{n\rightarrow\infty}\frac{1}{n}\sum^{n}_{i=1}\Big(\frac{1}{T}\sum^{T}_{t=1}\mathbb{E}\ln f(y_{it}\mid x_{it}, \theta, \alpha_{i}(\theta))\Big),$$
where $$\alpha_{i}(\theta)=\argmax_{\alpha\in\Gamma_{\alpha}}\frac{1}{T}\sum^{T}_{t=1}\mathbb{E}(\ln f(y_{it}\mid x_{it}, \theta, \alpha)).$$ 
With fixed $T$, $\widehat{\alpha}_{i}(\theta)\neq \alpha_{i}(\theta)$ in general. Therefore, $\theta_{T}\neq \theta_{0}$.   

To illustrate the problem, suppose $y_{it}$ has the normal distribution with mean $\alpha_{i0}$ and variance $\theta_{0}$, and the goal is to estimate $\theta_{0}$. The fixed effect estimator is  
$\widehat{\theta}=\frac{1}{nT}\sum^{n}_{i=1}\sum^{T}_{t=1}(X_{it}-\widehat{\alpha}_{i})^{2},$ where $\widehat{\alpha}_{i}=\frac{1}{T}\sum^{T}_{t=1}X_{it}.$ When $T$ is fixed and $n$ approaches infinity, \cite{NeymanScott1948} show that 
$$\widehat{\theta}\xrightarrow{p}\theta_{0}-\frac{\theta_{0}}{T}.$$
On the other hand, when $T$ also grows to infinity, the bias term $-\frac{\theta_{0}}{T}$ approaches zero. The large--$T$ panel literature generalizes this insight and shows that the incidental parameter problem becomes an asymptotic bias problem when $n$ and $T$ grow at the same rate. 
\section{The Indirect Fixed Effect Estimator} \label{Section: Overview}
The key feature of the indirect fixed effect estimator is to match $\widehat{\theta}$ with a fixed effect estimator from simulated data generated by $\widehat{\alpha}_{i}$'s and a given $\theta$. To avoid confusion, it is necessary to introduce notations to distinguish parameters in the simulation world from those in Section~\ref{Section: Model and FE}. More specifically, this paper uses $\beta$ and $\gamma_{i}$ to denote the vector of parameters of interest and individual effects in the log likelihood function using simulated data. 

To clarify the notations and introduce the implementation of indirect fixed effect estimator, this section first revisits the Neyman--Scott example. Using the panel  model as a concrete example, this section then illustrates the challenges associated with the presence of nonsmoothness and discusses the general estimation procedures.
\subsection{Neyman--Scott Example Revisited}
If $y_{it}\mid \alpha_{i0} \sim \mathcal{N}(\alpha_{i0}, \theta_{0})$ is i.i.d over $n$ and $t$, then the DGP of the observed data is:
\begin{equation*} 
    y_{it}(\alpha_{i0}, \theta) = \alpha_{i0}+\sqrt{\theta}u_{it}, \quad u_{it}\sim\mathcal{N}(0, 1).
\end{equation*}
This equation cannot be simulated without information on the distribution of $\alpha_{i0}$'s. The indirect fixed effect estimator uses $\widehat{\alpha}_{i}$'s instead, and the simulated data have the following DGP: 
\begin{equation*} 
    y^{h}_{it}(\widehat{\alpha}_{i}, \theta) = \widehat{\alpha}_{i}+\sqrt{\theta}u^{h}_{it}, \quad u^{h}_{it}\sim\mathcal{N}(0, 1),
\end{equation*}
where the superscript $h$ denotes a simulation path. The fixed effect estimator using $\{y^{h}_{it}(\widehat{\alpha}_{i}, \theta)\}$ is
\begin{equation*} 
    \widehat{\beta}^{h}(\theta, \boldsymbol{\widehat{\alpha}}) :=\frac{1}{nT}\sum^{n}_{i=1}\sum^{T}_{t=1}(y^{h}_{it}(\widehat{\alpha}_{i}, \theta)-\widehat{\gamma}_{i})^{2}=\frac{\theta}{nT}\sum^{n}_{i=1}\sum^{T}_{t=1}(u^{h}_{it}-\frac{1}{T}\sum^{T}_{t=1}u^{h}_{it})^{2},
\end{equation*} 
where $\boldsymbol{\widehat{\alpha}} := (\widehat{\alpha}_{1},\dots,\widehat{\alpha}_{n})$ and  $\widehat{\gamma}_{i}:=\frac{1}{T}\sum^{T}_{t=1}y^{h}_{it}(\widehat{\alpha}_{i}, \theta)$. The interpretation of $\widehat{\beta}^{h}(\theta, \boldsymbol{\widehat{\alpha}})$ is that the estimator changes if a different value of $\theta$ is used to simulate the data. Note that the $\widehat{\alpha}_{i}$'s are fixed throughout the simulation process. The indirect fixed effect estimator $\widetilde{\theta}$ is the solution to
$$\widehat{\theta}=\widehat{\beta}^{h}(\widetilde{\theta}, \boldsymbol{\widehat{\alpha}}).$$ 
Figure~(\ref{fig:neyman-scott}) illustrates the issues with $\widehat{\theta}$ and the performance of $\widetilde{\theta}$ in this example. The density of the fixed effect estimator $\widehat{\theta}$ conveys two messages: (1) fixed effect estimator is subject to a large bias and (2) a confidence interval around $\widehat{\theta}$ would not have the correct coverage. The density of $\widetilde{\theta}$ illustrates that (1) the new estimator corrects the bias significantly and (2) a confidence interval around $\widetilde{\theta}$ would have a much larger coverage than the one around $\widehat{\theta}$.  
\begin{figure}[H]
\centering
\caption{Comparison of FE and IFE}
\label{fig:neyman-scott}
\includegraphics[width=0.9\columnwidth]{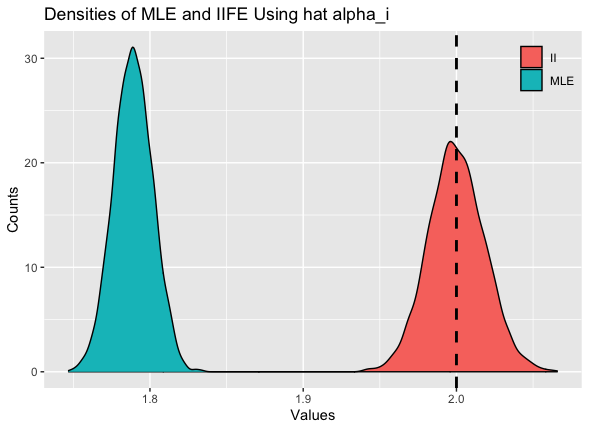}
\begin{minipage}{0.85\textwidth}
	{\footnotesize Note: Density plots of fixed effects and indirect fixed effect estimator for $\theta_{0}$. The DGP is $y_{it}=\alpha_{i0}+\sqrt{\theta_{0}}u_{it}$, where $u_{it}\sim\mathcal{N}(0, 1).$ The true value $\theta_{0}=2$ is depicted by the dashed line and $\alpha_{i0} = i$ for $i=1,\dots,n$. The sample size is $n = 2500, T = 5$ and number of simulation $H$ is set to be 1. The simulations are conducted 5000 times.\par}
\end{minipage}
\end{figure}

\begin{remark}[Caveat] \label{Remark: caveat} \rm
Due to the simple structure of the data, $\widehat{\theta}$ and $\widehat{\beta}^{h}(\theta, \boldsymbol{\widehat{\alpha}})$ have closed--form expressions, and $\widehat{\beta}^h(\theta,\boldsymbol{\widehat{\alpha}}) = \widehat{\beta}^h(\theta)$, i.e. the fixed effect estimator with the simulated data does not depend on the fixed effects. However, it is not generally the case in model~(\ref{structural model}). 
\end{remark} 
\subsection{Challenges due to Nonsmoothness with Incidental Parameters}
Consider the binary choice panel  model as a concrete example. Given $\theta$, $\widehat{\alpha}_{i}$'s and $x_{it}$, the simulated dependent variable is
$$y_{it}^{h}(\theta, \widehat{\alpha}_{i})=\boldsymbol{1}(x'_{it}\theta+\widehat{\alpha}_{i}>u^{h}_{it}), \quad u^{h}_{it}\sim\mathcal{N}(0, 1),$$
where $u^{h}_{it}$ are simulation draws from the standard normal distribution. The corresponding log likelihood function is 
\begin{equation} \label{simulated ll for }
    \frac{1}{nT}\sum^{n}_{i=1}\sum^{T}_{t=1}y_{it}^{h}(\theta, \widehat{\alpha}_{i})\log\Big(\Phi(x'_{it}\beta+\gamma_{i})\Big)+(1-y_{it}^{h}(\theta, \widehat{\alpha}_{i}))\log\Big(1-\Phi(x'_{it}\beta+\gamma_{i})\Big).
\end{equation}
It illustrates the three different aspects of simulated fixed effect estimator $\widehat{\beta}^{h}(\theta, \boldsymbol{\widehat{\alpha}})$. Because $y^{h}_{it}(\theta, \widehat{\alpha}_{i})$ is discontinuous in $\theta$ and $\widehat{\alpha}_{i}$'s, equation~(\ref{simulated ll for }) is discontinuous, which carries over to its maximizer $\widehat{\beta}^{h}(\theta, \boldsymbol{\widehat{\alpha}})$. In addition, estimating $\widehat{\beta}^{h}(\theta, \boldsymbol{\widehat{\alpha}})$ involves incidental parameters $\gamma_{i}$'s. The population version of equation~(\ref{simulated ll for }) does not have randomness due to data sampling, use of simulations and $\widehat{\alpha}_{i}$'s, and takes the form
\begin{equation} \label{limiting ll for }
    \frac{1}{nT}\sum^{n}_{i=1}\sum^{T}_{t=1}\mathbb{E}\Big[y_{it}(\theta, \alpha_{i0})\log\Big(\Phi(x'_{it}\beta+\gamma_{i})\Big)+(1-y_{it}(\theta, \alpha_{i0}))\log\Big(1-\Phi(x'_{it}\beta+\gamma_{i})\Big)\Big],
\end{equation}
where the expectation is over $u^{h}_{it}$ and $x_{it}$, and $\widehat{\alpha}_{i}$'s are replaced by $\alpha_{i0}$'s. 

\noindent \begin{remark}[A comparison with panel quantile regression (QR) models] \label{Remark: comparison with QR} \rm
One important type of nonlinear panel models that is not included in model~(\ref{structural model}) but also features nonsmoothness is panel QR models.\footnote{See \cite{GalvaoKato2018} for a recent survey.} \cite{KatoGalvaoMontesRojas2012} considered the following QR model with individual effects:
\begin{equation*}
    Q_{\tau}(y_{it}\mid x_{it}, \gamma_{i0}(\tau)) = \gamma_{i0}(\tau) + x_{it}'\beta_{0}(\tau),
\end{equation*}
where $\tau\in(0, 1)$ is a quantile index, and $Q_{\tau}(y_{it} \mid x_{it}, \gamma_{i0}(\tau))$ is the conditional $\tau$--quantile of $y_{it}$ given $(x_{it}, \gamma_{i0}(\tau))$. 
The fixed effects quantile regression (FE--QR) estimator for this model is
\begin{equation*}
    (\widehat{\gamma}_{\text{FE-QR}}, \widehat{\beta}_{\text{FE-QR}}) \in \argmin \frac{1}{nT}\sum^{n}_{i=1}\sum^{T}_{t=1}\rho_{\tau}(y_{it} - \gamma_{i} - x_{it}'\beta),
\end{equation*}
where $\gamma:=(\gamma_{1},\dots,\gamma_{n})'$ and $\rho_{\tau}(u):=\{\tau - \boldsymbol{1}\{u\leq 0\}\}u$ is the check function. The FE--QR estimator $\widehat{\beta}_{\text{FE-QR}}$ also contains estimated individual fixed effects and is non--smooth with respect to $\beta$ or $\gamma_{i}$'s because the check function is not differentiable. However, the estimator proposed in this paper is different in two aspects. First, the criterion function~(\ref{simulated ll for }) is still smooth in $\beta$ and $\gamma_{i}$'s. Instead, the source of non--smoothness comes from data simulations. Second, in panel QR models, the DGP does not involve estimated individual fixed effects. Therefore, the theoretical results in \cite{KatoGalvaoMontesRojas2012} do not apply here.
\end{remark}
\subsection{General Estimation Procedures}
\setcounter{example}{0}
From the known distribution $F_{u}$, the simulated unobservables $\{u^{h}_{it}\}$ are independently drawn for $h=1,\dots,H$, where $H$ denotes the number of simulated panel data sets. For a given value of $\theta$, let $y^{h}_{it}(\theta, \widehat{\alpha}_{i})$ denote the simulated dependent variable for simulation path $h$, then the sample log likelihood function using the $h$--th simulated data is
\begin{equation} \label{eqn: criterion in simulation}
    \frac{1}{nT}\sum^{n}_{i=1}\sum^{T}_{t=1}\ln f(y^{h}_{it}(\theta, \widehat{\alpha}_{i}) \mid x_{it}; \beta, \gamma_{i}),
\end{equation}
where $\beta$ and $\gamma_{i}$ respectively denote the finite--dimensional parameter and incidental parameter in the simulation world. The fixed effect estimator to this problem is 
\begin{equation*} 
    \widehat{\beta}^{h}(\theta, \boldsymbol{\widehat{\alpha}})=\argmax_{\beta\in\mathbb{R}^{p}}\frac{1}{nT}\sum^{n}_{i=1}\sum^{T}_{t=1}\ln f(y^{h}_{it}(\theta, \widehat{\alpha}_{i}) \mid x_{it}; \beta, \widehat{\gamma}_{i}(\beta, \theta, \widehat{\alpha}_{i})),
\end{equation*}
where 
$$\widehat{\gamma}_{i}(\beta, \theta, \widehat{\alpha}_{i})=\argmax_{\gamma\in\mathbb{R}}\frac{1}{T}\sum^{T}_{t=1}\ln f(y^{h}_{it}(\theta, \widehat{\alpha}_{i}) \mid x_{it}; \beta, \gamma).$$
Repeating this estimation for all the simulated panel data, the following average can be computed: 
$$\widehat{\beta}_{H}(\theta, \boldsymbol{\widehat{\alpha}}):=\frac{1}{H}\sum^{H}_{h=1}\widehat{\beta}^{h}(\theta, \boldsymbol{\widehat{\alpha}}),$$
and the indirect fixed effect estimator $\widetilde{\theta}^{H}$ is the solution to
\begin{equation} \label{eqn: estimation equation}
    \widehat{\theta}=\widehat{\beta}_{H}(\widetilde{\theta}^{H}, \boldsymbol{\widehat{\alpha}}),
\end{equation}
where the superscript $H$ stresses that the finite--sample performance depends on the number of simulations conducted. The box below summarizes the steps required to compute the estimator.
\begin{algorithm}[H]
\centering
\caption*{\textbf{Algorithm:} Computing the indirect fixed effect estimator} \label{alg: indirect--FE}
\begin{algorithmic}[1]
\item[(i)] Obtain $(\widehat{\theta}, \widehat{\alpha}_{1},\dots,\widehat{\alpha}_{n})$ using the observed data.
\vskip 0.2cm
\item[(ii)] Set a random seed and $H$. For each $i$ and $t$, draw unobservables $\{u^{h}_{it}\}^{H}_{h=1}$ from $F_{u}$.
\vskip 0.2cm
\item[(iii)] Given $\theta$ and $\widehat{\alpha}_{i}$'s, use model~(\ref{structural model}) and $\{u^{h}_{it}\}$ to simulate dependent variable
$\{y^{h}_{it}(\theta, \widehat{\alpha}_{i})\}$; construct data $\{y^{h}_{it}(\theta, \widehat{\alpha}_{i}), x_{it}\}$, where $i=1,\dots,n$ and $t=1,\dots,T$.
\vskip 0.2cm
\item[(iv)] Obtain $\widehat{\beta}^{h}(\theta, \boldsymbol{\widehat{\alpha}})$ using the simulated data in Step (iii).
\vskip 0.2cm
\item[(v)] Repeat steps (ii) and (iii) for all $h=1,\dots,H$ and solve for equation~(\ref{eqn: estimation equation}).
\end{algorithmic}
\end{algorithm}
\begin{remark}[Common random number] \label{Remark: CRN} \rm
Step $(ii)$ follows the standard practice for simulations \citep{GlassermanYao1992} by drawing unobserved shocks at the beginning of the algorithm. It implies that $\widehat{\beta}^{h}(\theta, \boldsymbol{\widehat{\alpha}})$ and  $\widehat{\beta}^{h'}(\theta, \boldsymbol{\widehat{\alpha}})$ are independent for $h\neq h'$ conditional on $x_{it}$. 
\end{remark}

\begin{remark} [The role of $H$] \label{Remark: the role of H} \rm
The number $H$ affects the finite--sample performance of the estimator, and increasing $H$ reduces the asymptotic variance. Just like SMM and indirect inference, there is a trade off between precision of the estimator and intensity of computation. The estimation method, however, is different from the simulated maximum likelihood \citep{ManskiLerman1981}, which is inconsistent for fixed $H$ due to a nonlinear transformation of simulated choice probabilities. 
\end{remark}

\section{Asymptotic Properties} \label{Section: Asymptotics}
This section starts with a discussion of the main assumptions that lead to theoretical properties of $\widehat{\theta}$ and $\widehat{\alpha}_{i}$'s. These assumptions are standard in large--$T$ panel data models \citep{HahnKuersteiner2011}, and they also impose certain structures that help establish the asymptotic properties of the indirect fixed effect estimator. Additional assumptions are imposed to ensure the simulations do not affect the panel data structure.

\begin{asu}[Large $T$ asymptotics] \label{largeT}
$n, T\rightarrow\infty$ such that $nT^{-1}\rightarrow\kappa\in(0, \infty).$
\end{asu}
Assumption~\ref{largeT} requires that the time series dimension grows at the same rate as the cross section dimension. The assumption defines the large--$T$ asymptotics framework and allows to transform the incidental parameter problem from a consistency to a bias problem, the latter of which can be quantified. From a practitioner's point of view, if the ratio $T/n$ is not negligible, then it is reasonable to consider the large T asymptotics.

\begin{asu} [Sampling of observed data] \label{assumption on data}
 \subasu \label{cross-section independence}
$\{z_{it}\}_{t=1}^{\infty}$ are independent across $i$; 
\subasu \label{alpha mixing}
For each $i$, $\{z_{it}\}_{t=1}^{\infty}$ is a stationary $\alpha$--mixing sequence with mixing coefficient $\alpha_{i}(m)$ such that
$\sup_{i}\lvert \alpha_{i}(m)\rvert\leq Ka^{m}$ for some $a$ such that $0<a<1$ and some $K>0$.
\end{asu}
Assumption~\ref{cross-section independence} imposes independence along the cross--section dimension. Assumption~\ref{alpha mixing} imposes a weak temporal dependence on each individual time series. The quantity $\alpha_{i}(m)$ measures for each $i$ how much dependence exists between data separated by at least $m$ time periods, and a uniform bound is imposed so as to bound covariances and moments when using law of large numbers (LLN) and central limit theorem (CLT). Interested readers can refer to Section 3.4 in \cite{White2000} for definitions and properties. Note that Assumption~\ref{assumption on data} rules out non--stationary explanatory variables such as time effects and linear trends. 

\begin{asu}[Identification] \label{identification}
Denote $G_{(i)}(\theta, \alpha_{i})\equiv \frac{1}{T}\sum^{T}_{t=1}\mathbb{E}\Big[\ln f(y_{it}\mid x_{it}; \theta, \alpha_{i})\Big]$. For each $\eta>0$, 
$$\inf_{i}\Big[G_{(i)}(\theta_{0}, \alpha_{i0})-\sup_{(\theta, \alpha): \lVert (\theta, \alpha) - (\theta_{0}, \alpha_{i0})\rVert\geq\eta}G_{(i)}(\theta, \alpha)\Big]>0.$$
\end{asu}
Assumption~\ref{identification} is a sufficient condition that ensures the log likelihood function admits a unique maximizer based on the time series variation. This assumption allows to prove the consistency of fixed effect estimators under large--T asymptotics. The indirect fixed effect estimator also requires this assumption for consistency.

\begin{asu}[Envelope condition] \label{Lipschitz}
\subasu \label{compactness} The parameter $\varphi_{i}:=(\theta, \alpha_{i})\in \text{int } \Theta\times\Gamma_{\alpha},$ where $\Theta$ and $\Gamma_{\alpha}$ are compact, convex subsets of $\mathbb{R}^{p}$ and $\mathbb{R}$ respectively. \\
\subasu \label{envelope} There exists an envelope function $M(z_{it})$ such that  
$$\lVert D^{\nu} \ln f(y_{it}\mid x_{it}; \varphi_{1})- D^{\nu} \ln f(y_{it}\mid x_{it}; \varphi_{2})\rVert\leq M(z_{it})\lVert \varphi_{1} - \varphi_{2}\rVert,$$ 
where $\nu:=(\nu_{1},\dots,\nu_{p+1})$ is a vector of non--negative integers $\nu_{j}$ and $\lvert\nu\rvert=\sum^{p+1}_{j=1}\nu_{j}$. In addition, for $\ln f(y_{it}\mid x_{it}; \varphi):=\ln f(y_{it}\mid x_{it}; \theta, \alpha_{i})$,
$$D^{\nu} \ln f(y_{it}\mid x_{it}; \varphi):=\partial^{\lvert\nu\rvert}\ln f(y_{it}\mid x_{it}; \varphi)/(\partial\varphi_{1}^{\nu_{1}}\dots\partial\varphi_{p+1}^{\nu_{p+1}})$$ 
and $\lvert\nu\rvert\leq 5$, the function $M(z_{it})$ satisfies 
$$\sup_{\varphi_{i}\in int\Theta\times\Gamma_{\alpha}}\lVert D^{\nu}\ln f(y_{it}\mid x_{it}; \varphi)\rVert\leq M(z_{it})$$ and $\sup_{i}\mathbb{E}(\lvert M(z_{it})\rvert^{(10+10q)/(1-10\tilde{\nu})+\delta})<\infty$ for some integer $q\geq p/2+2$, $\delta>0$ and $0<\tilde{\nu}<1/10$.
\end{asu}
Assumption~\ref{compactness} imposes compactness of parameter space, which is standard for establishing asymptotic properties of extremum estimators \citep{NeweyMcFadden1994}. Compactness is convenient for proving uniform convergence with nonsmooth criterion functions \citep{Newey1991}. Assumption~\ref{envelope} imposes a Lipschitz condition on the log likelihood function and a moment condition on the envelope function. This allows to establish uniform law of large number (ULLN) of sample log likelihood function and hence the pointwise consistency of $\widehat{\theta}$.  

Under these assumptions and some regularity conditions on the Hessian matrix, \cite{HahnKuersteiner2011} established the following two results:
\begin{equation} \label{uniform consistency of alphai}
    \max_{1\leq i\leq n}\lvert \widehat{\alpha}_{i}-\alpha_{i0}\rvert = o_{p}(1),
\end{equation}
\begin{equation} \label{consistency of fe}
    \widehat{\theta} = \theta_{0} + \frac{A(\theta_{0}, \boldsymbol{\alpha_{0}})}{\sqrt{nT}} + \frac{B(\theta_{0}, \boldsymbol{\alpha_{0}})}{T} + o_{p}\Big(\frac{1}{T}\Big),
\end{equation}
where $\boldsymbol{\alpha_{0}}:=(\alpha_{10},\dots,\alpha_{n0})$. Equation~(\ref{uniform consistency of alphai}) states that the maximal deviation of $\widehat{\alpha}_{i}$ from $\alpha_{i0}$ converges to zero. This uniform consistency result is crucial for the theory of indirect fixed effect estimator because it justifies the usage of $\widehat{\alpha}_{i}$'s for data simulations. Equation~(\ref{consistency of fe}) characterizes the asymptotic relationship between $\widehat{\theta}$ and $\theta_{0}$. The term $A(\theta_{0}, \boldsymbol{\alpha_{0}})$ is the influence function that satisfies the central limit theorem (CLT) with mean zero. The term $B(\theta_{0}, \boldsymbol{\alpha_{0}})$ converges to its expected value. Therefore, $\widehat{\theta}$ is consistent, asymptotically normal, but biased. \cite{HahnKuersteiner2011} derived the analytical forms of both terms, which are complicated functions of $\theta_{0}$ and $\boldsymbol{\alpha_{0}}$. 

Because the new estimator involves simulations, the following regularity condition is required so that the simulated data still maintain the mixing properties. Another regularity condition is that the parameter space in the simulation world is compact. Because $(\beta, \gamma)$ is just a change of notation from $(\theta, \alpha)$, this assumption is natural. 

\begin{asu} [Simulation] \label{ass: simulation}
\subasu \label{simulation mixing} 
Assumption~\ref{assumption on data} holds for the simulated process for all $\theta\in\Theta$. 
\subasu \label{compactness in simulation} 
The parameter spaces for $\beta$ and $\gamma_{i}$, $\Theta_{\beta}$ and $\Gamma_{\gamma}$ are compact. 
\end{asu}
In sum, Assumption~\ref{ass: simulation} allows for an asymptotic representation of simulated fixed effect estimator $\widehat{\beta}^{h}(\theta, \widehat{\boldsymbol{\alpha}})$ that resembles the one for $\widehat{\theta}$, i.e., equation~(\ref{consistency of fe}). 
\begin{remark}[Inferring $\widetilde{\theta}^{H}$ from fixed effect estimators] \label{Remark: invertibility} \rm
Backing out $\widetilde{\theta}^{H}$ from $\widehat{\beta}_{H}(\widetilde{\theta}^{H}, \boldsymbol{\widehat{\alpha}})$ requires an invertible relationship $\theta \mapsto \beta(\theta, \boldsymbol{\alpha_{0}})$, where $\beta(\theta, \boldsymbol{\alpha_{0}})$ is the maximizer of the limiting function for equation~(\ref{eqn: criterion in simulation})
\begin{equation*} 
    \frac{1}{nT}\sum^{n}_{i=1}\sum^{T}_{t=1}\mathbb{E}\Big[\ln f(y_{it}(\theta, \alpha_{i0}) \mid x_{it}; \beta, \gamma_{i})\Big].
\end{equation*}
The expectation is taken over simulation draws and sampling of observed data, and $\widehat{\alpha}_{i}$ is replaced with $\alpha_{i0}$. This function is essentially the population function for the fixed effects estimation problem, except that $\beta$ and $\gamma_{i}$ are used to denote parameters for estimation in the simulation world. Assumption~\ref{identification} thus ensures the uniqueness of $\beta(\theta)$, which is $\theta$. As such, invertibility is satisfied trivially.\footnote{For readers who are familiar with indirect inference, the relationship means that the binding function is an identity. This is because the auxiliary model is identical to the structural model, and thus the parameters in the two models coincide. Many papers that employ indirect inference often have to assume invertibility of the binding function \citep{CollardWexler2013, GospodinovKomunjerNg2017}, but this assumption is guaranteed in this paper.} Although $\beta(\theta)$ is an identity function, for the rest of the paper this notation is kept to avoid the confusion between maximum of the limit and a parameter for data generation. 
\end{remark}
\subsection{Consistency}
In order for the indirect inference--type estimator to be consistent, three conditions should be satisfied \citep{GourierouxMonfortRenault1993}: an invertible relationship between $\theta$ and $\beta(\theta, \boldsymbol{\alpha_{0}})$, pointwise convergence of $\widehat{\theta}$ to $\beta(\theta_{0}, \boldsymbol{\alpha_{0}})$, and uniform convergence of $\widehat{\beta}^{h}(\theta, \boldsymbol{\widehat{\alpha}})$ to $\beta(\theta, \boldsymbol{\alpha_{0}})$ over the compact parameter space $\Theta$. The first condition is satisfied because $\beta(\theta, \boldsymbol{\alpha_{0}})$ is an identity, and equation~(\ref{consistency of fe}) gives the second condition. The following proposition states the uniform convergence condition.

\begin{proposition}[Uniform convergence of fixed effect estimator using simulated data] \label{Proposition: uniform consistency} Under Assumptions~\ref{largeT}--\ref{ass: simulation}, 
    $$\sup_{\theta\in\Theta}\lVert\widehat{\beta}^{h}(\theta,\boldsymbol{\widehat{\alpha}})-\beta(\theta, \boldsymbol{\alpha_{0}})\rVert\xrightarrow{p}0.$$
\end{proposition}

The current proof specializes to panel  models, but it is generalizable to other models that feature concavity and smoothness in $(\beta, \gamma_{i})$. Details are available in Appendix~\ref{Appendix: Main Theorem Proof}, and here the main ideas are discussed. 

Proving the uniform convergence condition with nonsmoothness requires two steps: pointwise convergence of $\widehat{\beta}^{h}(\theta,\boldsymbol{\widehat{\alpha}})$ to $\beta(\theta, \boldsymbol{\alpha_{0}})$, and a stochastic equicontinuity condition as follows:
\begin{equation} \label{stoch equi for uniform consistency}
    \mathbb{E}\Big(\sup_{\lVert \theta_{1}-\theta_{2}\rVert\leq \delta}\lVert \widehat{\beta}^{h}(\theta_{1},\boldsymbol{\widehat{\alpha}})-\widehat{\beta}^{h}(\theta_{2},\boldsymbol{\widehat{\alpha}})\rVert\Big)\leq C\delta,
\end{equation}
where $C$ is a constant and $\delta$ is a positive scalar. 

Following the standard argument in \cite{NeweyMcFadden1994}, pointwise convergence requires a ULLN result of log likelihood function using simulated data~(\ref{simulated ll for }) to the limiting log likelihood~(\ref{limiting ll for }). The log likelihood~(\ref{simulated ll for }) has two sources of randomness: the first source comes from sampling variation of observed data, and the other is from simulations of unobservables. The non--standard part, however, is that data are simulated using $\widehat{\alpha}_{i}$'s. Therefore, it is necessary to first show that (\ref{simulated ll for }) uniformly well approximates the log likelihood using data generated by $\alpha_{i0}$'s:
\begin{equation} \label{eqn: ll using true alphas}
\frac{1}{nT}\sum^{n}_{i=1}\sum^{T}_{t=1}\int_{U}\Big[y_{it}^{h}(\theta, \alpha_{i0})\log\Big(\Phi(x'_{it}\beta+\gamma_{i})\Big)+(1-y_{it}^{h}(\theta, \alpha_{i0}))\log\Big(1-\Phi(x'_{it}\beta+\gamma_{i})\Big)\Big]dF_{u},
\end{equation}
where the integration is with respect to the distribution of simulation draws $u^{h}_{it}$ to eliminate randomness from simulations. The details are available in Lemma~\ref{uniform convergence}, and intuition is provided here. Because panel data are independent along the cross section, it suffices to show that each individual's log likelihood:
$$ \frac{1}{T}\sum^{T}_{t=1}y_{it}^{h}(\theta, \widehat{\alpha}_{i})\log\Big(\Phi(x'_{it}\beta+\gamma_{i})\Big)+(1-y_{it}^{h}(\theta, \widehat{\alpha}_{i}))\log\Big(1-\Phi(x'_{it}\beta+\gamma_{i})\Big),$$
satisfies this property. Given $\theta$, this individual log likelihood is an additive and multiplicative combination of indicator functions of scalar $\widehat{\alpha}_{i}$ and smooth functions of $(\beta, \gamma_{i})$, which belongs to classes of functions that satisfy stochastic equicontinuity \citep{VanDerVaartWellner1996}. Therefore, its empirical process: 
\begin{align*}
    \nu_{T}&(\alpha_{i})=\frac{1}{T}\sum^{T}_{t=1}\Big[ y_{it}^{h}(\theta, \alpha_{i})\log\Big(\Phi(x'_{it}\beta+\gamma_{i})\Big)+(1-y_{it}^{h}(\theta, \alpha_{i}))\log\Big(1-\Phi(x'_{it}\beta+\gamma_{i})\Big) \\
    & - \int_{U}\Big(y_{it}^{h}(\theta,\alpha_{i})\log\Big(\Phi(x'_{it}\beta+\gamma_{i}) \Big) +(1-y_{it}^{h}(\theta, \alpha_{i}))\log \Big(1-\Phi(x'_{it}\beta+\gamma_{i})\Big)\Big)dF_{u}\Big],
\end{align*}
is stochastic equicontinuous. Combined with uniform consistency result of $\widehat{\alpha}_{i}$'s and LLN of $\nu_{T}(\alpha_{i0})$, an application of the triangular inequality leads to the uniform approximation result. Now that (\ref{eqn: ll using true alphas}) only has randomness from observed data, its uniform convergence to the limiting log likelihood~(\ref{limiting ll for }) follows the argument as in \cite{HahnKuersteiner2011}. As such, the pointwise convergence of $\widehat{\beta}^{h}(\theta, \boldsymbol{\widehat{\alpha}})$ follows through.\footnote{Details are available in Lemma~\ref{point-wise consistency}.} To verify the stochastic equicontinuity condition~(\ref{stoch equi for uniform consistency}), note that the profiled log likelihood:
\begin{align*} 
    \widehat{Q}(\beta; \theta) = \frac{1}{nT}\sum^{n}_{i=1}&\sum^{T}_{t=1}\Big[y^{h}_{it}(\theta, \widehat{\alpha}_{i})\log\Big(\Phi(x_{it}'\beta + \widehat{\gamma}_{i}(\beta))\Big) \\
    &+ (1-y^{h}_{it}(\theta, \widehat{\alpha}_{i}))\log\Big(1-\Phi(x_{it}'\beta+ \widehat{\gamma}_{i}(\beta))\Big)\Big]
\end{align*}
is concave in $\beta$. By definition, $\widehat{\beta}^{h}(\theta_{1}, \boldsymbol{\widehat{\alpha}})$ satisfies $\partial\widehat{Q}(\widehat{\beta}^{h}(\theta_{1}, \boldsymbol{\widehat{\alpha}}); \theta_{1})/\partial\beta=0.$ A first--order Taylor expansion with respect to $\widehat{\beta}^{h}(\theta_{1}, \boldsymbol{\widehat{\alpha}})$ around $\widehat{\beta}^{h}(\theta_{2}, \boldsymbol{\widehat{\alpha}})$ and positive--definiteness of the Hessian shows that $\widehat{\beta}^{h}(\theta_{1}, \boldsymbol{\widehat{\alpha}})-\widehat{\beta}^{h}(\theta_{2}, \boldsymbol{\widehat{\alpha}})$ is bounded by $$\Big\lvert\frac{\partial\widehat{Q}(\widehat{\beta}^{h}(\theta_{2}, \boldsymbol{\widehat{\alpha}}); \theta_{2})}{\partial\beta}-\frac{\partial\widehat{Q}(\widehat{\beta}^{h}(\theta_{2}, \boldsymbol{\widehat{\alpha}}); \theta_{1})}{\partial\beta}\Big\rvert,$$ 
which, by the Cauchy--Schwarz inequality, is bounded by the product of two terms: a smooth function of $(\beta, \gamma_{i})$ and 
\begin{equation} \label{eqn: diff of indicators}
    \frac{1}{nT}\sum^{n}_{i=1}\sum^{T}_{t=1}(y^{h}_{it}(\theta_{1}, \widehat{\alpha}_{i})-y^{h}_{it}(\theta_{2}, \widehat{\alpha}_{i})).
\end{equation}
Therefore, it suffices to bound the two terms in expectation. The technical challenge mainly comes from proving this for equation~(\ref{eqn: diff of indicators}) that features nonsmooth components. Although indicator functions are well--known to have controlled complexities \citep{Andrews1994}, and a similar result on the difference of indicator functions with univariate variable is given in \cite{ChenLintonVanKeilegom2003}, here $\theta$'s can be multi--dimensional. It turns out that the expectation of (\ref{eqn: diff of indicators}) satisfies the $L^{2}$--smoothness regularity condition in \cite{Andrews1994}.\footnote{See proof of Proposition~\ref{Proposition: uniform consistency} for details.} Therefore, the stochastic equicontinuity condition for $\widehat{\beta}^{h}(\theta, \boldsymbol{\widehat{\alpha}})$ is verified.

Armed with Proposition~\ref{Proposition: uniform consistency}, the consistency of $\widetilde{\theta}^{H}$ follows the arguments as in \cite{GourierouxMonfortRenault1993}. The proof is straightforward because there is no need to consider a weighting matrix. 
\begin{theorem}[Consistency of indirect fixed effect estimator] \label{consistency}
Under Assumptions~\ref{largeT}--\ref{ass: simulation}, 
$$\widetilde{\theta}^{H}\xrightarrow{p}\theta_{0}.$$
\end{theorem}
\subsection{Bias Correction and Asymptotic Normality}
Recall that the indirect fixed effect estimator using $H$ simulations $\widetilde{\theta}^{H}$ is the solution to $\widehat{\theta}=\widehat{\beta}_{H}(\widetilde{\theta}^{H}, \boldsymbol{\widehat{\alpha}})$.  Non--differentiability of $\theta\mapsto\widehat{\beta}_{H}(\theta, \boldsymbol{\widehat{\alpha}})$ means that the techniques in the indirect inference literature \citep[e.g.,][]{GourierouxMonfortRenault1993} are not applicable. The following stochastic equicontinuity assumption is imposed. 

\begin{asu} \label{asu: replace smooth}
For all positive deterministic sequences $\delta_{nT}\downarrow 0$, 
\begin{align*}
    \sup_{\lVert\theta_{1}-\theta_{2}\rVert\leq\delta_{nT}}\sqrt{nT}\lVert\widehat{\beta}_{H}(\theta_{1}, \boldsymbol{\widehat{\alpha}})&-\widehat{\beta}_{H}(\theta_{2}, \boldsymbol{\widehat{\alpha}})-\mathbb{E}(\widehat{\beta}_{H}(\theta_{1}, \boldsymbol{\widehat{\alpha}})-\widehat{\beta}_{H}(\theta_{2}, \boldsymbol{\widehat{\alpha}}))\rVert \xrightarrow{p} 0.
\end{align*}
\end{asu}

Assumption~\ref{asu: replace smooth} requires that the difference between $\widehat{\beta}_{H}(\theta_{1}, \boldsymbol{\widehat{\alpha}})$ and $\widehat{\beta}_{H}(\theta_{2}, \boldsymbol{\widehat{\alpha}})$ can be approximated by its expectation at a $\sqrt{nT}$ rate. Combined with consistency of $\widetilde{\theta}^{H}$ and the mean value theorem, it allows to analyze the asymptotic normality of $\widetilde{\theta}^{H}$ through the lens of fixed effect estimators as follows:
\begin{equation} \label{eqn: transformation of asymptotic analysis}
    \sqrt{nT}(\widetilde{\theta}^{H}-\theta_{0})=\sqrt{nT}(\widehat{\theta}-\widehat{\beta}_{H}(\theta_{0}, \boldsymbol{\widehat{\alpha}}))+o_{p}(1).
\end{equation}

Recall that equation~(\ref{consistency of fe}) characterizes the representation of $\widehat{\theta} - \theta_{0}$. Because the same regression is run on simulated data $h$ and the likelihood is smooth in $(\beta, \gamma_{i})$, the same structure of representation arises, namely that
\begin{equation} \label{eqn: stochastic expansion in simulation world}
    \widehat{\beta}^{h}(\theta_{0}, \boldsymbol{\widehat{\alpha}}) - \theta_{0} = \frac{A^{h}(\theta_{0}, \boldsymbol{\widehat{\alpha}})}{\sqrt{nT}} + \frac{B^{h}(\theta_{0}, \boldsymbol{\widehat{\alpha}})}{T} + o_{p}\Big(\frac{1}{T}\Big).
\end{equation}
The terms $A^{h}(\theta_{0}, \boldsymbol{\widehat{\alpha}})$ and $B^{h}(\theta_{0}, \boldsymbol{\widehat{\alpha}})$ reflect that the data are generated using $\theta_{0}$, $\boldsymbol{\widehat{\alpha}}$ and simulated unobservables $\{u^{h}_{it}\}$. A combination of (\ref{consistency of fe}), (\ref{eqn: transformation of asymptotic analysis}) and (\ref{eqn: stochastic expansion in simulation world}) therefore leads to 
\begin{align}
    \sqrt{nT}(\widetilde{\theta}^{H} - \theta_{0}) = \Big(A(\theta_{0}, \boldsymbol{\alpha_{0}}) & - \frac{1}{H}\sum^{H}_{h=1}A^{h}(\theta_{0}, \boldsymbol{\widehat{\alpha}})\Big) \nonumber \\
    & + \sqrt{\frac{n}{T}}\Big(B(\theta_{0}, \boldsymbol{\alpha_{0}}) - \frac{1}{H}\sum^{H}_{h=1}B^{h}(\theta_{0}, \boldsymbol{\widehat{\alpha}})\Big) + o_{p}(1).
\end{align}
This equation reflects two observations. First, $\widehat{\theta}$ is unbiased if $B(\theta_{0}, \boldsymbol{\alpha_{0}})$ and $B^{h}(\theta_{0}, \boldsymbol{\widehat{\alpha}})$ both converge to the same limit. Second, $\widehat{\theta}$ is asymptotically normal if $A(\theta_{0}, \boldsymbol{\alpha_{0}})$ and $A^{h}(\theta_{0}, \boldsymbol{\widehat{\alpha}})$ converge to the same limiting distribution, but the variance is inflated by a factor of $1/H$. The rest of the section provides the main ideas of the proof.  

The intuition can be gained by setting $H=1$ and considering an infeasible fixed effect estimator $\widehat{\beta}_{H}(\theta_{0}, \boldsymbol{\alpha_{0}})$, which is obtained from data simulated by $(\theta_{0}, \alpha_{0})$. Then the representation of $\widehat{\beta}_{H}(\theta_{0}, \boldsymbol{\alpha_{0}})-\theta_{0}$ takes the form
\begin{align*}
     \sqrt{nT}(\widehat{\beta}_{H}(\theta_{0}, \boldsymbol{\alpha_{0}}) - \theta_{0}) = A^{h}(\theta_{0}, \boldsymbol{\alpha_{0}}) + \sqrt{\frac{n}{T}} B^{h}(\theta_{0}, \boldsymbol{\alpha_{0}}) + o_{p}(1).
\end{align*}
The theory of indirect inference implies that $B(\theta_{0}, \boldsymbol{\alpha_{0}})$ and $B^{h}(\theta_{0}, \boldsymbol{\alpha_{0}})$ converge to the same probability limit. Because the actual simulated data are generated by $\widehat{\alpha}_{i}$'s, it suffices to show that $B^{h}(\theta_{0}, \widehat{\alpha})$ uniformly well approximates $B^{h}(\theta_{0}, \boldsymbol{\alpha_{0}})$ such that the approximation error is asymptotically negligible. More specifically, the bias term using simulated data takes the following form, 
$$B^{h}(\theta_{0}, \boldsymbol{\widehat{\alpha}}) = -\Big[\frac{1}{n}\sum^{n}_{i=1}\mathcal{I}_{i}(\theta_{0}, \widehat{\alpha}_{i})\Big]^{-1} \frac{1}{n}\sum^{n}_{i=1}B^{h}_{i}(\theta_{0}, \widehat{\alpha}_{i}),$$ 
where $\mathcal{I}_{i}(\theta_{0}, \widehat{\alpha}_{i})$ is individual $i$'s information matrix, and it is a smooth function of all its arguments. Therefore, $\mathcal{I}_{i}(\theta_{0}, \widehat{\alpha}_{i})\xrightarrow{p}\mathcal{I}_{i}(\theta_{0}, \alpha_{i0})$ for each $i$. Each $B^{h}_{i}(\theta_{0}, \widehat{\alpha}_{i})$ is nonsmooth in $\widehat{\alpha}_{i}$, and the following assumption is imposed such that $B^{h}(\theta_{0}; \boldsymbol{\widehat{\alpha}})$ replaces $B^{h}(\theta_{0}, \boldsymbol{\alpha_{0}})$ with negligible errors:
\begin{asu} (Restricting changes due to using $\widehat{\alpha}_{i}$) \label{asu: tech for bias}
$$\mathbb{E}\max_{1\leq i\leq n}\Big\lvert B^{h}_{i}(\theta_{0}, \widehat{\alpha}_{i}) - B^{h}_{i}(\theta_{0}, \alpha_{i0})\Big\rvert^{2} = o_{p}(1)$$
\end{asu}

\begin{proposition}[Bias correction of $\widetilde{\theta}^{H}$] \label{prop: bias correction} Under Assumptions~\ref{largeT}--\ref{asu: tech for bias},
    $$\lvert B^{h}(\theta_{0}, \boldsymbol{\widehat{\alpha}}) - B^{h}(\theta_{0}, \boldsymbol{\alpha_{0}})\rvert\xrightarrow{p}0.$$
\end{proposition}
\noindent As such, the indirect fixed effect estimator corrects the bias. 
\cite{HahnKuersteiner2011} derived the analytical expression of the term $A(\theta_{0}, \boldsymbol{\alpha_{0}})$. The term $A^{h}(\theta_{0}, \boldsymbol{\widehat{\alpha}})$ has the same structure, namely 
$$A^{h}(\theta_{0}, \boldsymbol{\widehat{\alpha}}) = \Big[\frac{1}{n}\sum^{n}_{i=1}\mathcal{I}^{h}_{i}(\theta_{0}, \widehat{\alpha}_{i})\Big]^{-1}\frac{1}{\sqrt{nT}}\sum^{n}_{i=1}\sum^{T}_{t=1}U^{h}_{it}(\theta_{0}, \widehat{\alpha}_{i}),$$ 
where $A^{h}_{it}(\theta_{0}, \widehat{\alpha}_{i})$ is a combination of high--order derivatives of the log likelihood. The following high--level assumption is imposed:
\begin{asu} \label{assumption for normality}
$\frac{1}{\sqrt{nT}}\sum^{n}_{i=1}\sum^{T}_{t=1} U^{h}_{it}(\theta_{0}, \widehat{\alpha}_{i})=\frac{1}{\sqrt{nT}}\sum^{n}_{i=1}\sum^{T}_{t=1}U^{h}_{it}(\theta_{0}, \alpha_{i0})+o_{p}(1).$
\end{asu}

Under this assumption, $A^{h}(\theta_{0}, \boldsymbol{\widehat{\alpha}})$ can uniformly well approximate $A^{h}(\theta_{0}, \boldsymbol{\alpha_{0}})$ with negligible errors, and the asymptotic normality result in indirect inference literature follows through \citep[][Proposition 5]{GourierouxMonfortRenault1993}. Combined with Proposition~\ref{prop: bias correction}, the indirect fixed effect estimator is asymptotically unbiased and normal. 


\begin{theorem} \label{CLT}
Under Assumptions~\ref{largeT}--\ref{assumption for normality},
$$\sqrt{nT}\Big(\widetilde{\theta}^{H}-\theta_{0}\Big)\xrightarrow{d}\mathcal{N}\Big(0, \Big(1+\frac{1}{H}\Big)\Omega\Big),$$
where $\Omega:=\mathbb{E}(A(\theta_{0}, \boldsymbol{\alpha_{0}})A(\theta_{0}, \boldsymbol{\alpha_{0}})').$
\end{theorem}
The estimation of the variance--covariance matrix uses the estimated Hessian matrix of the sample log likelihood function from the real data. As previously discussed in Remark~\ref{Remark: the role of H}, the number of simulations $H$ shows up as a factor that inflates the asymptotic variance. There are two interpretations. The first is in line with other simulation--based methods: using simulations introduces an additional source of uncertainty and it is manifested through an increase in variance. The other interpretation is related to the trade--off between bias and variance. Because the indirect inference estimator debiases fixed effects estimator, the variance is larger, and it is quantified by the number of panel data simulated. 

\cite{DhaeneJochmans2015} proposed a half--panel method that removes the leading bias. Intuitively, the method splits in half the panel along the time series dimension, obtains the fixed effect estimators for the half samples and applies a linear combination with respect to the full--sample fixed effect estimator. Theoretically, the method does not change the asymptotic variance because the influence function is linear. However, in finite samples the variance is inflated due to an inefficient use of data. The indirect fixed effect method is explicit about the bias--variance tradeoff in that the asymptotic variance is multiplied by $H$. However, as will be shown in the simulations and applications, this method leads to a smaller standard error compared to the half--panel bias correction method. 
\section{Application to Female Labor Force Participation} \label{Section: Application}
Research on the relationship between female labor force participation and fertility is complicated by the presence of unobserved factors that affect both decisions. Following \cite{hyslop99}, this paper addresses the omitted variable issue by including individual fixed effects into the binary response panel  model for the female labor force participation. 

The data come from the Panel Study of Income Dynamics (PSID) and constitute a nine--year longitudinal sample spanning from 1979 to 1988. The sample includes 664 women aged 18–-60 in 1985 who were continuously married with husbands in the labor force in each of the sample periods and changed their labor force participation statuses. Consider the following static specification:
\begin{equation*} 
    y_{it}=\boldsymbol{1}\{x'_{it}\theta+\alpha_{i}>u_{it}\}, \quad u_{it}\sim\mathcal{N}(0, 1),
\end{equation*}
where $y_{it}$ denotes the labor force participation indicator for woman $i$ at time $t$, and $x_{it}$ denotes a vector of time--varying covariates. These covariates include numbers of children of at most 2 years of age, between 3 and 5 years of age, between 6 and 17 years of age; log of the husband's income,\footnote{This variable serves as a proxy for permanent nonlabor income \citep{hyslop99}.} age and age squared. The individual effects $\alpha_{i}$'s are included to control for time--invariant unobserved heterogeneity such as willingness to work or ability.

\begin{table}[htbp!]
\centering
\begin{threeparttable}
\caption{Parameter Estimates for Static LFP} \label{empirics}
\begin{tabular}{rrrrrrr}
  \hline\hline
  & kids0\_2 & kids3\_5 & kids6\_17 & loghusinc & age & age2 \Tstrut\Bstrut\\ 
  \hline
FE  & -0.71 & -0.41 & -0.13 & -0.24 & 2.32 & -0.29 \Tstrut\Bstrut\\ 
    & (0.06) & (0.05) & (0.04) & (0.05) & (0.38) & (0.05)\vspace{0.08cm}\\ 
IFE--1 & -0.65 & -0.36 & -0.08 & -0.17 & 2.24 & -0.29 \\
    & (0.08) & (0.07) & (0.06) & (0.08) & (0.53) & (0.07)\vspace{0.08cm}\\    
IFE--10 & -0.60 & -0.32 & -0.10 & -0.30 & 2.08 & -0.27 \\
    & (0.06) & (0.05) & (0.04) & (0.06) & (0.39) & (0.05)\vspace{0.08cm}\\
IFE--20 & -0.60 & -0.32 & -0.10 & -0.30 & 2.08 & -0.27 \\
    & (0.06) & (0.05) & (0.04) & (0.06) & (0.38) & (0.05)\vspace{0.08cm}\\
ABC & -0.63 &  -0.37 &  -0.11 & -0.22 & 2.39 &-0.25 \\
    & (0.06) & (0.05) & (0.04) & (0.05) & (0.38) & (0.05)\vspace{0.08cm}\\
BC--HN & -0.62 & -0.36 &  -0.10 & -0.21 & 1.73 &  -0.22 \\
 & (0.06) & (0.05) & (0.04) & (0.05) & (0.38) & (0.05)\vspace{0.08cm}\\
HBC &  -0.92 & -0.58 & -0.26 & -0.30 & 2.28 & -0.26  \\
 & (0.09) & (0.09) & (0.08) & (0.07) & (0.89) & (0.12)\vspace{0.08cm}\\
\hline
\bottomrule
\end{tabular}
\begin{tablenotes}[flushleft]
\linespread{1}\footnotesize
\item\hspace*{-\fontdimen2\font}\textit{Notes:} Standard errors are reported in the parenthesis and are computed based on the Hessian matrix of profiled log likelihood. The HBC estimates and standard errors computation follows page 1025 in \cite{DhaeneJochmans2015}. 
\end{tablenotes}
\end{threeparttable}   
\end{table}

Table~(\ref{empirics}) reports estimates of index coefficients using different methods. The standard errors are reported in parentheses. The standard errors for the fixed effects are computed from the Hessian of the profiled log likelihood. IFE--1, IFE--10 and IFE--20 denote indirect fixed effect estimators with $H$ being 1, 10 and 20 respectively, and their respective standard errors are computed by multiplying the FE standard errors by $(1+\frac{1}{H})$. For comparisons, the table includes results using half--panel jackknife method (HBC) \citep{DhaeneJochmans2015}, analytical bias correction (ABC) \citep{Fernandez-Val2009} and the leave--one--out jackknife method (BC--HN) \citep{HahnNewey2004}. The ABC has the same standard errors as the uncorrected fixed effect estimators, while the standard error computation for BC--HN and HBC follow the descriptions in \cite{HahnNewey2004} and \cite{DhaeneJochmans2015} respectively. The results show that the uncorrected estimates of index coefficients are about 15\% larger (in absolute value) than their bias-corrected counterparts, indirect fixed effect estimators are closely comparable to ABC and BC--HN, and HBC produces estimates that are larger in magnitude. Because HBC achieves bias correction through sample splitting, the standard errors are larger. 
\section{Monte Carlo Simulations} \label{Section: Monte Carlo}
This section considers Monte Carlo simulations calibrated to the same PSID data. The details of calibration procedures are available in Appendix~\ref{appendix: calibration}. The indirect inference fixed effect estimator is compared with the fixed effect estimation, the ABC and two jackknife bias correction methods. All simulations are done 1000 times and $H$ is set to 10. The coverage reports the proportion of the times that $\theta_{0}$ falls within the 95\% confidence interval. All of the other statistics are relative to the true parameters and multiplied by 100. 
\begin{table}[htbp!]
\centering
\begin{threeparttable}
\caption {\label{tab:MC static} Simulation Results for Static LFP} 
\begin{tabular}{rrrrrrrrrrrrr}
\hline\hline
  & \multicolumn{3}{c}{FE} &  & \multicolumn{3}{c}{IFE--10} & \multicolumn{4}{c}{IFE--20}  \Tstrut\Bstrut \\ 
\cline{2-4} \cline{6-8} \cline{10-12} 
            &  Bias   & Std Dev  &  Cvge & &  Bias  & Std Dev  &  Cvge & &  Bias  & Std Dev  &  Cvge\Tstrut\Bstrut\\    \hline
 kids0\_2   &  14.75  & 9.62 & 0.79 & & -4.78  & 8.18 & 0.94  & & -4.21 & 8.80 & 0.93\Tstrut\Bstrut\\       
 kids3\_5   &  14.74  & 14.27 & 0.91 & & -6.83  & 13.12 & 0.95 & & -6.43 & 13.23 & 0.94\Tstrut\Bstrut\\         
 kids6\_17  &  14.49  & 36.58    & 0.94 & & -18.06 & 38.81 & 0.94 & & -16.98 & 39.83 & 0.93\Tstrut\Bstrut\\
 loghusinc  &  14.87  & 25.83    & 0.94 & & -3.34  & 26.76 & 0.97 & & -4.66 & 24.97 & 0.96\Tstrut\Bstrut\\
 age        &  13.53  & 19.34    & 0.92 & & 0.24   & 7.79  & 0.97 & & -0.87 & 9.98 & 0.98\Tstrut\Bstrut\\
 age2       &  13.47  & 20.61    & 0.92 & & -2.25  & 22.51    & 0.96 & & -4.00 & 23.00 & 0.96\Tstrut\Bstrut\\
 \hline
 \bottomrule
\end{tabular}
\begin{tablenotes}[flushleft]
\linespread{1}\footnotesize
\item\hspace*{-\fontdimen2\font}\textit{Notes:} FE denotes fixed effects estimates. IFE--10 and IFE--20 denote indirect fixed effect estiamtes with $H$ being 10 and 20. Cvge denotes the empirical coverage probability. The nominal coverage is 95\%. Simulations are conducted 1000 times, and all relative statistics are multiplied by 100. The nominal coverage is 95\%.
\end{tablenotes}
\end{threeparttable}  
 \end{table}
 
 \begin{table}[htbp!]
\centering
\begin{threeparttable}
\caption {\label{tab:MC static 2} Simulation Results for Static LFP} 
\begin{tabular}{rrrrrrrrrrrrr}
\hline\hline
  & \multicolumn{3}{c}{ABC} &  & \multicolumn{3}{c}{BC--HN} & \multicolumn{4}{c}{HBC}   \Tstrut\Bstrut\\ 
\cline{2-4} \cline{6-8} \cline{10-12} 
&  Bias   & Std Dev  &  Cvge & &  Bias  & Std Dev  &  Cvge & &  Bias  & Std Dev  &  Cvge\Tstrut\Bstrut\\    \hline
 kids0\_2   &  1.18  & 8.39 & 0.95 & & -3.46  & 8.09  & 0.95  & & -5.22 & 12.47 & 0.96\Tstrut\Bstrut\\       
 kids3\_5   &  1.35  & 12.59 & 0.96 & & -3.33  & 12.17 & 0.96 & & -4.70 & 21.12 & 0.98\Tstrut\Bstrut\\         
 kids6\_17  &  1.54  & 32.39    & 0.95 & & -3.50 & 31.12 & 0.96 & & -4.16 & 56.48 & 0.98\Tstrut\Bstrut\\
 loghusinc  &  1.55  & 22.75 & 0.96 & & -3.43  & 21.87 & 0.96 & & -6.11 & 28.27 & 0.98\Tstrut\Bstrut\\
 age &  0.38  & 27.50  & 0.97 & & 4.27   & 16.67     & 0.96 & & -4.02 & 34.81 & 0.98\Tstrut\Bstrut\\
 age2  &  0.48  & 18.41 & 0.96 & & -4.36  & 17.78 & 0.94 & & -3.95 & 37.07 & 0.98\Tstrut\Bstrut\\
 \hline
 \bottomrule
\end{tabular}
\begin{tablenotes}[flushleft]
\linespread{1}\footnotesize
\item\hspace*{-\fontdimen2\font}\textit{Notes:} ABC denotes analytical bias correction in \cite{Fernandez-Val2009}. BC--HN denotes leave--one--out jackknife bias correction in \cite{HahnNewey2004}. HBC denotes split--panel bias correction in \cite{DhaeneJochmans2015}. Cvge denotes the empirical coverage probability. The nominal coverage is 95\%. Simulations are conducted 1000 times, and all relative statistics are multiplied by 100. The nominal coverage is 95\%.
\end{tablenotes}
\end{threeparttable}  
 \end{table}
 
Table~(\ref{tab:MC static}) reports the simulation results of fixed effects and indirect fixed effect estimators. Fixed effect estimators are subject to a bias that is of the same order of magnitude as the standard deviation. This leads to severe under--coverage of the confidence intervals. The indirect fixed effect estimators, on the other hand, reduce bias by a margin without much inflation in the standard deviation. Therefore, the empirical coverage is close to the nominal value of $95\%$.


Table~(\ref{tab:MC static 2}) tabulates the simulation results of ABC and two jackknife bias correction methods. Compared with IFE, ABC features smaller biases as it removes the bias term based on a plugged--in estimate, but standard deviations are comparable. Turning to the other two methods that automatically correct bias, first note that BC--HN admits smaller biases and standard deviations than HBC. The simulation results are in line with those reported in \cite{HughesHahn2020}, who theoretically showed that HBC has a larger higher--order variance and remaining bias than BC--HN.\footnote{Therefore, in practice it is recommended to use panel bootstrap to obtain standard errors for HBC.} On the other hand, IFE is comparable with BC--HN in terms of both bias and standard deviation. A theoretical exploration is left for future work. 

The current theory is restricted to strictly exogenous explanatory variables, but Monte Carlo simulations in Appendix~\ref{appendix: dynamic simulation} shows that the method can accommodate lagged dependent variables as well. Naturally, the next step is to extend the current theory to allow for dynamics in the DGP. 

Like other bias correction methods, the theoretical properties of the indirect fixed effect estimator are predicated on the large--$T$ assumption. Therefore, to compare how the estimator performs against other methods under varying lengths of time periods, this paper follows the literature \citep[e.g.,][]{HahnNewey2004, Fernandez-Val2009, HughesHahn2020} and considers the following simulation design: 
\begin{align*}
    & y_{it} = \boldsymbol{1}\{\theta_{0}x_{it} + \alpha_{i} - \varepsilon_{it}\geq 0\}, \quad \theta_{0}=1, \quad \alpha_{i}\sim\mathcal{N}(0, 1), \quad \varepsilon_{it}\sim\mathcal{N}(0, 1); \\
    & x_{it} = t/10 + x_{i,t-1}/2 + u_{it}, \quad x_{i0} = u_{i0}, \quad u_{it}\sim U(-0.5, 0.5).
\end{align*}
\begin{align*}
    & y_{it} = \boldsymbol{1}\{\theta_{0}x_{it} + \alpha_{i} - \varepsilon_{it}\geq 0\}, \quad \theta_{0}=1, \quad \alpha_{i}\sim\mathcal{N}(0, 1), \quad \varepsilon_{it}\sim\mathcal{N}(0, 1); \\
    & x_{it} = t/10 + x_{i,t-1}/2 + u_{it}, \quad x_{i0} = u_{i0}, \quad u_{it}\sim U(-0.5, 0.5).
\end{align*}
The numerical experiments consider panels with $n = \{100, 200\}$ and $T = \{4, 8, 12\}$.

\begin{table}[htbp!]
\centering
\begin{adjustbox}{max width=\textwidth}
\begin{threeparttable}
\caption {\label{tab:MC varying T} Estimates of $\theta_{0}$} 
\begin{tabular}{rrrrrrrrrrrrrrrrrr}
\hline\hline
  & \multicolumn{3}{c}{FE} &  & \multicolumn{3}{c}{IFE--1} & \multicolumn{4}{c}{IFE--10}  & \multicolumn{4}{c}{HBC} \Tstrut\Bstrut\\ 
\cline{2-4} \cline{6-8} \cline{10-12} \cline{14-16}
            &  Bias   & Std Dev  &  Coverage & &  Bias  & Std Dev  &  Coverage & &  Bias  & Std Dev  &  Coverage & &  Bias  & Std Dev  &  Coverage \Tstrut\Bstrut\\    \hline\Tstrut\Bstrut
 $n=100, T=4$   &  39.81  & 39.69 & 0.87 & & -3.00 & 44.06 & 0.99  & & -2.98 & 38.52 & 0.96 & & -42.41 & 80.30 & 0.94\vspace{0.1cm}\\       
 $n=100, T= 8$  &  18.58  & 14.27 & 0.86 & & 1.45 & 17.33 & 0.98 & & 0.94  & 12.98 & 0.96  & & -7.00 & 25.09 & 0.96\vspace{0.1cm}\\         
 $n=100, T = 12$  &  12.93  & 9.89    & 0.85 & & 0.29 & 11.17 & 0.99 & & 0.45 & 10.12 & 0.95 & & -3.50 & 16.55 & 0.97\vspace{0.1cm}\\
 $n=200, T = 4$  &  41.21  & 27.17    & 0.78 & & 1.67  & 30.70    & 0.99 & & 1.17 & 27.97 & 0.97 & & -36.56 & 53.53 & 0.92\vspace{0.1cm}\\
 $n=200, T = 8$  &  18.25  & 10.31 & 0.76 & & 1.17   & 12.84     & 0.96 & & 0.83 & 10.85 & 0.96 & & -6.11 & 17.94 & 0.96\vspace{0.1cm}\\
 $n=200, T = 12$ &  13.34  & 6.92  & 0.73 & & 1.47  & 8.29    & 0.95 & & 1.25 & 6.93 & 0.95 & & -2.38 & 11.65 & 0.97\vspace{0.1cm}\\
 \hline
 \bottomrule
\end{tabular}
\begin{tablenotes}[flushleft]
\linespread{1}\footnotesize
\item\hspace*{-\fontdimen2\font}\textit{Notes:} FE denotes fixed effects estimates. IFE--1 and IFE--10 denote indirect fixed effect estiamtes with $H$ being 1 and 10 respectively. HBC denotes split--sample jackknife bias correction in \cite{DhaeneJochmans2015}. Cvge denotes the empirical coverage probability. The nominal coverage is 95\%. Simulations are conducted 1000 times, and all the statistics are multiplied by 100. The nominal coverage is 95\%.
\end{tablenotes}
\end{threeparttable}  
\end{adjustbox}
 \end{table}

Table~(\ref{tab:MC varying T}) reports the simulation results. The fixed effect estimators are subject to large biases, even when the time periods is 12. Because their biases are comparable to the standard deviations, fixed effect estimators exhibit undercoverage, which implies under--rejections. HBC has a poor performance when $T=4$, and this is because each of the split sample only uses 2 time periods for estimation. HBC substantially reduces the bias when $T$ is 8 or 12, but it is subject to a large dispersion, which reflects a larger confidence interval. As a result, the empirical coverage of HBC is larger than the nominal value of 95\%.

The indirect fixed effect estimator has a better bias reduction performance compared to HBC, especially when $T=4$. This means that the new estimator is less sensitive to the time periods, and thus can be potentially useful for short--panel applications as well. The number of simulation paths affects the dispersion. When $H=1$, the estimator has a larger dispersion compared to the fixed effect estimator. When $H=10$, the dispersion is reduced. The trade--off for setting a large $H$ in practice is an increase in computation time. For example, when $n=200$ and $T=12$, it took roughly sixteen minutes to obtain the results for $H=10$ with ten cores on the MacBook Pro (M1, 2020). For $H=1$, it took about two minutes. It is interesting to explore algorithms that efficiently search for solutions to non--smoonth functions, but this is beyond the scope of the paper.

It is worth noting that both the application and the simulation design feature non--stationary regressors. The results provide suggestive evidence that the indirect fixed effect estimator can accommodate these variables, which do not satisfy the stationarity condition in Assumption~\ref{assumption on data}. However, the theoretical exploration is beyond the scope of this paper. In another ongoing project, the author proposes a new method called \textit{crossover jackknife} that deals with non--stationarity explicitly.

\section{Conclusion} \label{Section: Conclusion}
Fixed effect estimations of nonlinear panel models are subject to large biases of point estimates and incorrect coverages of confidence intervals. This paper proposes a new estimator that reduces the bias and obtains standard errors without bootstrap. 

There are at least three other questions for further explorations. First, average partial effects are often the quantities of interest in nonlinear models. This paper establishes theoretical properties of finite dimensional parameters, and it could be interesting to explore if they can be extended to handle average partial effects, which is a function of explanatory variables, parameters of interest and incidental parameters.

Second, this paper directly works with non--smooth log likelihood function and establishes the asymptotic properties of the new estimator. However, a practical concern of non--smoothness is that gradient--based optimization schemes cannot be used for estimation, and gradient--free schemes like Nelder--Mead face computational difficulty in high--dimensional problems. The indirect fixed effect estimator might benefit from approaches like kernel smoothing, but the theoretical justification can be nontrivial as smoothing can introduce an additional bias.


Finally, incorporating unobserved heterogeneity into the dynamic discrete choice (DDC) models is an active area of research. One popular approach treats unobserved heterogeneity as an unobserved state variable and assumes individuals can be categorized into a finite number of types \citep{KasaharaShimotsu2009, ArcidiaconoMiller2011}. Introducing fixed effects circumvents the need to take a stand on the number of types, but can potentially complicate identification and estimation: the individual effects show up in both the current payoff and the continuation value, the latter of which has to be solved using a fixed--point algorithm. It would be exciting to investigate whether some of the ideas in this paper can be applied to incorporate fixed effects into DDC models.

\newpage
\bibliographystyle{ecta}
\bibliography{reference}
\newpage
\appendix
\section{Auxiliary Results} \label{Appendix: Auxiliary Lemmas}
\setcounter{equation}{0}
\numberwithin{equation}{section}
\subsection{Proof of Lemma~\ref{uniform convergence}}
\begin{lemma}[Uniform Convergence of Sample Criterion Function using Simulated Data] \label{uniform convergence}
$$\max_{1\leq i\leq n}\sup_{(\beta, \gamma)}\Big\lvert \widehat{G}^{h}_{(i)}(\beta, \gamma) - G_{(i)}(\beta, \gamma) \Big\rvert \xrightarrow{p} 0,$$
where
\begin{align*}
    \widehat{G}^{h}_{(i)}(\beta, \gamma) & = \frac{1}{T}\sum^{T}_{t=1}\ln f(y^{h}_{it}(\theta, \widehat{\alpha}_{i}) \mid x_{it}; \beta, \gamma); \\
    G_{(i)}(\beta, \gamma) & = \frac{1}{T}\sum^{T}_{t=1}\mathbb{E}\ln f(y_{it}(\theta, \alpha_{i0}) \mid x_{it}; \beta, \gamma). \\
\end{align*}
\end{lemma}
\begin{proof} The proof consists of two main steps. The first step deals with $\widehat{\alpha}_{i}$'s in data simulation and shows that $\widehat{G}^{h}_{(i)}$ is uniformly close to a criterion that uses $\alpha_{i0}$ to simulate the data, i.e., 
$$\widetilde{G}_{(i)}(\beta, \gamma)=\frac{1}{T}\sum^{T}_{t=1}\int_{U}\ln f(y_{it}(\theta, \alpha_{i0})\mid x_{it}; \beta, \gamma)dF(u).$$
The second step is a uniform law of large number results showing that $\widetilde{G}_{(i)}(\beta, \gamma)$ uniformly converges to $G_{(i)}(\beta, \gamma)$. 

\noindent \textbf{Step 1:} Given $\theta$ and a scalar $\tau$, note that 
\begin{align*}
    \frac{1}{T}\sum^{T}_{t=1}\ln f(y^{h}_{it}(\theta, \tau)\mid x_{it}; \beta, \gamma) := & \frac{1}{T}\sum^{T}_{t=1}y^{h}_{it}(\theta, \tau)\ln \Phi(x_{it}'\beta + \gamma) \\
    & + (1-y^{h}_{it}(\theta, \tau))\ln(1-\Phi(x_{it}'\beta + \gamma))
\end{align*}
consists of two components: (1) an indicator function of scalar $\tau$ and (2) a smooth, bounded and monotone function of $(\beta, \gamma)$. The indicator function $y^{h}_{it}(\theta, \widehat{\alpha}_{i})$ belongs to type I class of \cite{Andrews1994}, which satisfies Pollard's entropy condition. The second component belongs to a class of functions satisfying bracketing entropy condition \citep[][Section 2.7.2]{VanDerVaartWellner1996}. 

Because $\frac{1}{T}\sum^{T}_{t=1}\ln f(y^{h}_{it}(\theta, \tau)\mid x_{it}; \beta, \gamma)$ is an additive and multiplicative combination of the two classes of components, its function class also satisfies the entropy condition \citep{Andrews1994}, which is the primitive condition for stochastic equicontinuity. More specifically, define the following empirical process:
$$\nu_{T}(\tau)=\frac{1}{T}\sum^{T}_{t=1}\Big[\ln f(y^{h}_{it}(\theta, \tau)\mid x_{it}; \beta, \gamma) - \int_{U}\ln f(y^{h}_{it}(\theta, \tau)\mid x_{it}; \beta, \gamma)dF_{u}\Big],$$
where the integration is over the known distribution of simulation draws. 
By one of the equivalent definitions of stochastic equicontinuity \cite[i.e.,][p.2252]{Andrews1994}, the following condition holds: for every sequence of constants $\{\delta_{T}\}$ that converges to zero, 
\begin{equation} \label{eqn: replace individual effects time series}
    \sup_{(\beta, \gamma)\in\mathcal{B}\times\Gamma_{\gamma}, \lvert\tau_{1}-\tau_{2}\rvert\leq\delta_{T}}\sqrt{T}\lvert \nu_{T}(\tau_{1})-\nu_{T}(\tau_{2})\rvert \xrightarrow{p} 0.
\end{equation}
A first--order Taylor expansion on $\int_{U}\ln f(y^{h}_{it}(\theta, \alpha_{i0})\mid x_{it}; \beta, \gamma)$ with respect to $\alpha_{i0}$ around $\widehat{\alpha}_{i}$ yields
\begin{align*}
    \int_{U}\ln f(y^{h}_{it}(\theta, \alpha_{i0})\mid x_{it}; \beta, \gamma)dF_{u}=&\int_{U}\ln f(y^{h}_{it}(\theta, \widehat{\alpha}_{i})\mid x_{it}; \beta, \gamma)dF_{u} \\
    & +\frac{\partial\int_{U}\ln f(y^{h}_{it}(\theta, \overline{\alpha}_{i})\mid x_{it}; \beta, \gamma)dF_{u}}{\partial\alpha_{i}}(\widehat{\alpha}_{i}-\alpha_{i0}).
\end{align*}
Combined with condition~(\ref{eqn: replace individual effects time series}),
\begin{align*}
    \sqrt{T}\Big\lvert\frac{1}{T}&\sum^{T}_{t=1}\Big[\ln f(y^{h}_{it}(\theta, \widehat{\alpha}_{i}) \mid x_{it}; \beta, \gamma) - \int_{U}\ln f(y^{h}_{it}(\theta, \alpha_{i0})\mid x_{it}; \beta, \gamma)dF_{u}\Big]\Big\rvert \\
    & = \sqrt{T}\Big\lvert\nu_{T}(\widehat{\alpha}_{i}) - \frac{\partial\int_{U}\ln f(y^{h}_{it}(\theta, \overline{\alpha}_{i})\mid x_{it}; \beta, \gamma)dF_{u}}{\partial\alpha_{i}}(\widehat{\alpha}_{i}-\alpha_{i0}) ]\Big\rvert \\
    & \leq \sqrt{T}\lvert \nu_{T}(\widehat{\alpha}_{i})\rvert + \sqrt{T} \Big\lvert\frac{1}{T}\sum^{T}_{t=1} \frac{\partial\int_{U}\ln f(y^{h}_{it}(\theta, \overline{\alpha}_{i})\mid x_{it}; \beta, \gamma)dF_{u}}{\partial\alpha_{i}}(\widehat{\alpha}_{i}-\alpha_{i0}) \Big\rvert \\
    & = \sqrt{T}\lvert \nu_{T}(\alpha_{i0})+\nu_{T}(\widehat{\alpha}_{i})-\nu_{T}(\alpha_{i0})\rvert \\
    & +\sqrt{T} \Big\lvert\frac{1}{T}\sum^{T}_{t=1} \frac{\partial\int_{U}\ln f(y^{h}_{it}(\theta, \overline{\alpha}_{i})\mid x_{it}; \beta, \gamma)dF_{u}}{\partial\alpha_{i}}(\widehat{\alpha}_{i}-\alpha_{i0}) \Big\rvert\\
    & \leq \sqrt{T}\lvert \nu_{T}(\alpha_{i0})\rvert+\sqrt{T}\lvert\nu_{T}(\widehat{\alpha}_{i})-\nu_{T}(\alpha_{i0})\rvert \\
    & + \sqrt{T}\Big\lvert\frac{1}{T}\sum^{T}_{t=1} \frac{\partial\int_{U}\ln f(y^{h}_{it}(\theta, \overline{\alpha}_{i})\mid x_{it}; \beta, \gamma)dF_{u}}{\partial\alpha_{i}}\Big\rvert\cdot\lvert\widehat{\alpha}_{i}-\alpha_{i0}\rvert. 
\end{align*}
where the third and last lines are due to triangular inequality. Because $\nu_{T}(\alpha_{i0})$ is a normalized sum of mean zero random variables, $\nu_{T}(\alpha_{i0})\xrightarrow{p} 0$ by LLN. The second term is the stochastic equicontinuity condition in Eq.~(\ref{eqn: replace individual effects time series}). Because the derivative is bounded by Assumption~\ref{Lipschitz} and $\max_{1\leq i\leq n}\lvert\widehat{\alpha}_{i}-\alpha_{i0}\rvert=o_{p}(1)$ \citep[][Theorem 4]{HahnKuersteiner2011}, the third term is thus $o_{p}(1)$. Therefore
$$\sup_{(\beta, \gamma)\in\mathcal{B}\times\Gamma_{\gamma}}\Big\lvert\frac{1}{T}\sum^{T}_{t=1}[\ln f(y^{h}_{it}(\theta, \widehat{\alpha}_{i}) \mid x_{it}; \beta, \gamma) - \int_{U}\ln f(y^{h}_{it}(\theta, \alpha_{i0})\mid x_{it}; \beta, \gamma)dF_{u}]\Big\rvert \xrightarrow{p} 0.$$
\textbf{Step 2:} The second part of the proof shows that
$$\max_{1\leq i\leq n}\sup_{(\beta, \gamma)}\Big\lvert\widetilde{G}_{(i)}(\beta, \gamma) - G_{(i)}(\beta, \gamma) \Big\rvert \xrightarrow{p} 0.$$
Following the the proof structure of Lemma 4 in \cite{HahnKuersteiner2011}, note that 
\begin{align*}
     P\Big[&\max_{1\leq i\leq n}\sup_{(\beta, \gamma)}\Big\lvert \widetilde{G}_{(i)}(\beta, \gamma) - G_{(i)}(\beta, \gamma) \Big\rvert \geq \eta\Big] \leq \sum^{n}_{i=1}P\Big[\sup_{(\beta,\gamma)}\Big\lvert \widetilde{G}_{(i)}(\beta, \gamma) - G_{(i)}(\beta, \gamma) \Big\rvert \geq \eta\Big].
\end{align*}
Since the parameter space is compact, it suffices to show that 
$$\sup_{\Gamma_{j}}\Big\lvert \widetilde{G}_{(i)}(\beta, \gamma) - G_{(i)}(\beta, \gamma) \Big\rvert \rightarrow 0,$$
where $\Gamma_{j}$ is a subset of $\mathcal{B}\times\Gamma_{\gamma}$ such that $\lVert\beta - \beta'\rVert\leq \varepsilon$ and $\lvert \gamma - \gamma'\rvert\leq \varepsilon$ for $(\beta, \gamma)$ and $(\beta', \gamma')\in \Gamma_{j}.$
By Assumption~\ref{Lipschitz} on $G_{(i)}$, 
\begin{align*}
    & \Big\lvert G_{(i)}(\beta, \gamma) - G^{h}_{(i)}(\beta', \gamma') \Big\rvert \leq \mathbb{E}M(z_{it})\lvert (\beta, \gamma) - (\beta', \gamma') \rvert < \varepsilon \mathbb{E}M(z_{it}), \\
    & \Big\lvert \widetilde{G}_{(i)}(\beta, \gamma) - \widetilde{G}_{(i)}(\beta', \gamma') \Big\rvert \leq \frac{1}{T}\sum^{T}_{t=1}M(z_{it})\lvert (\beta, \gamma) - (\beta', \gamma') \rvert < \frac{\varepsilon}{T}\sum^{T}_{t=1}M(z_{it}).
\end{align*}
By the triangular inequality,
\begin{align*}
    \Big\lvert \widetilde{G}_{(i)}(\beta, \gamma) - G_{(i)}(\beta, \gamma) \Big\rvert &- \Big\lvert \widetilde{G}_{(i)}(\beta', \gamma') - G_{(i)}(\beta', \gamma') \Big\rvert \\
    &\leq \Big\lvert \Big(\widetilde{G}_{(i)}(\beta, \gamma)  - \widetilde{G}_{(i)}(\beta', \gamma')\Big) - \Big( G_{(i)}(\beta, \gamma) - G_{(i)}(\beta', \gamma')\Big)\Big\rvert \\
    & \leq \Big\lvert \widetilde{G}_{(i)}(\beta, \gamma)  - \widetilde{G}_{(i)}(\beta', \gamma' ) \Big\rvert + \Big\lvert G_{(i)}(\beta, \gamma) - G_{(i)}(\beta', \gamma')\Big\rvert \\
    & < \varepsilon \mathbb{E}M(z_{it}) + \frac{\varepsilon}{T}\sum^{T}_{t=1}M(z_{it}) \\
    & = \frac{\varepsilon}{T}\Big(\sum^{T}_{t=1}M(z_{it})-\mathbb{E}M(z_{it})\Big) + \frac{\varepsilon}{T}\mathbb{E}M(z_{it})+\varepsilon\mathbb{E}M(z_{it}) \\
    & < \frac{\varepsilon}{T}\Big\lvert \sum^{T}_{t=1}M(z_{it})-\mathbb{E}M(z_{it}) \Big\rvert + 2\varepsilon\mathbb{E}M(z_{it}).
\end{align*}
Therefore by a rearrangement of the terms, 
\begin{align*}
    \Big\lvert \widetilde{G}_{(i)}(\beta, \gamma) - G_{(i)}(\beta, \gamma ) \Big\rvert \leq \Big\lvert &\widetilde{G}_{(i)}(\beta', \gamma') - G_{(i)}(\beta', \gamma') \Big\rvert \\
    & + \frac{\varepsilon}{T}\Big\lvert \sum^{T}_{t=1}M(x_{it})-\mathbb{E}M(x_{it}) \Big\rvert + 2\varepsilon\mathbb{E}M(x_{it}). 
\end{align*}
Let $\varepsilon$ be such that $2\varepsilon\max_{i}\mathbb{E}M(z_{it})<\frac{\eta}{3}$, then 
\begin{align*}
    P\Big[&\sup_{\Gamma_{j}}\Big\lvert \widetilde{G}_{(i)}(\beta, \gamma) - G_{(i)}(\beta, \gamma) \Big\rvert>\eta\Big] \\
     \leq &P\Big[\Big\lvert \widetilde{G}_{(i)}(\beta', \gamma') - G_{(i)}(\beta', \gamma') \Big\rvert> \frac{\eta}{3}\Big] + P\Big[\frac{1}{T}\Big\lvert \sum^{T}_{t=1}M(x_{it})-\mathbb{E}M(x_{it}) \Big\rvert>\frac{\eta}{3\varepsilon}\Big] \\
     & + P\Big[2\varepsilon\mathbb{E}M(x_{it})>\frac{\eta}{3}\Big]\\
    =&o(T^{-2}),
\end{align*}
where the last line follows as the first two terms on the right--hand side are $o(T^{-2})$ by Lemma 1 in \cite{HahnKuersteiner2011} and the last term is of probability zero by construction. Since $n=O(T)$,
\begin{align*}
     P\Big[&\max_{1\leq i\leq n}\sup_{(\beta, \gamma)}\Big\lvert \widetilde{G}_{(i)}(\beta, \gamma) - G_{(i)}(\beta, \gamma) \Big\rvert \geq \eta\Big] \\
     &\leq \sum^{n}_{i=1}\sum^{m(\varepsilon)}_{j=1}P\Big[\sup_{\Gamma_{j}}\Big\lvert \widetilde{G}_{(i)}(\beta, \gamma) - G_{(i)}(\beta, \gamma) \Big\rvert\geq\eta\Big] \\
     &=o(T^{-1})
\end{align*}
\end{proof}
\subsection{Proof of Lemma~\ref{point-wise consistency}}
\begin{lemma}[Pointwise Consistency of Auxiliary Estimator in the Simulation World] \label{point-wise consistency}
$\forall \theta\in\Theta$,
$$\widehat{\beta}^{h}(\theta, \boldsymbol{\widehat{\alpha}}) \xrightarrow{p} \beta(\theta, \boldsymbol{\alpha_{0}}) = \theta.$$
\end{lemma}
\begin{proof}
The previous lemma shows that using $\boldsymbol{\widehat{\alpha}}$ for data simulation well approximates data simulated using $\boldsymbol{\alpha_{0}}$, therefore the randomness in the log likelihood function only comes from observed data. The proof structure of this lemma follows from that for Theorem 3 in \cite{HahnKuersteiner2011}, with minor modification of notations. Fix $\eta>0$ and set 
$$\varepsilon = \inf_{i}\Big[G_{(i)}(\theta, \alpha_{i0})-\sup_{\{(\beta, \gamma): \lVert(\beta, \gamma)-(\theta, \alpha_{i0})\rVert>\eta\}}G_{(i)}(\beta, \gamma)\Big]>0$$
With probability $1-o(T^{-1})$, 
\begin{align*}
    \max_{\lVert \beta - \theta\rVert>\eta, \gamma_{1},\dots,\gamma_{n}}\frac{1}{n}\sum^{n}_{i=1}\widehat{G}^{h}_{(i)}(\beta, \gamma_{i}) &\leq \max_{\lVert (\beta, \gamma_{i}) - (\theta, \alpha_{i0})\rVert>\eta}\frac{1}{n}\sum^{n}_{i=1}\widehat{G}^{h}_{(i)}(\beta, \gamma_{i}) \\
    &\leq \max_{\lVert (\beta, \gamma_{i}) - (\theta, \alpha_{i0}) \rVert>\eta}\frac{1}{n}\sum^{n}_{i=1}G_{(i)}(\beta, \gamma_{i}) + \frac{1}{3}\varepsilon \\ 
    & < \frac{1}{n}\sum^{n}_{i=1}G_{(i)}(\theta, \alpha_{i0}) - \frac{2}{3}\varepsilon \\
    & < \frac{1}{n}\sum^{n}_{i=1}\widehat{G}^{h}_{(i)}(\theta, \alpha_{i0}) - \frac{1}{3}\varepsilon,
\end{align*}
where the second and last inequalities are due to Lemma~\ref{uniform convergence}. By definition 
$$\max_{\beta, \gamma_{1},\dots,\gamma_{n}}\frac{1}{n}\sum^{n}_{i=1}\widehat{G}^{h}_{(i)}(\beta, \gamma_{i}) \geq \frac{1}{n}\sum^{n}_{i=1}G^{h}_{(i)}(\theta, \alpha_{i0}).$$
Hence
$$P\Big[\lVert \widehat{\beta}^{h}(\theta, \boldsymbol{\widehat{\alpha}}) - \beta(\theta, \boldsymbol{\alpha_{0}}) \rVert \geq \eta\Big]=o(T^{-1}).$$
\end{proof}
\section{Proofs of Main Results} \label{Appendix: Main Theorem Proof}
\setcounter{equation}{0}
\numberwithin{equation}{section}
\subsection{Proof of Proposition~\ref{Proposition: uniform consistency}}
\begin{proof}
The structure of the proof follows Theorem 1 in \cite{Newey1991}, which requires four  main pieces. The parameter space $\Theta$ is compact by assumption. The limiting function $\beta(\theta, \boldsymbol{\alpha_{0}})$ is continuous since it is an identity function. Lemma~\ref{point-wise consistency} establishes the pointwise convergence result using simulated data: $\forall \theta\in\Theta,$ $\widehat{\beta}^{h}(\theta, \boldsymbol{\widehat{\alpha}})\xrightarrow{p}\beta(\theta, \boldsymbol{\alpha_{0}})$. Therefore, it suffices to prove that $\widehat{\beta}^{h}(\theta, \boldsymbol{\widehat{\alpha}})$ is stochastic equicontinuous. This section uses $\widehat{\beta}^{h}(\theta)$ to ease the notation. 

\noindent By Markov inequality, $\forall \eta>0$,
$$Pr\Big(\sup_{\theta\in\Theta}\lVert\widehat{\beta}^{h}(\theta)-\beta(\theta, \boldsymbol{\alpha_{0}})\rVert>\eta\Big)\leq \frac{1}{\eta}\mathbb{E}\Big(\sup_{\theta\in\Theta}\lVert\widehat{\beta}^{h}(\theta)-\beta(\theta, \boldsymbol{\alpha_{0}})\rVert\Big).$$
Combined with the compactness assumption, it suffices to show that
\begin{equation} \label{eqn for stoc eq}
    \mathbb{E}\Big(\sup_{\lVert \theta_{1}-\theta_{2}\rVert\leq \delta}\lVert \widehat{\beta}^{h}(\theta_{1})-\widehat{\beta}^{h}(\theta_{2})\rVert\Big)\leq C\delta,
\end{equation}
where $\delta$ denotes a positive scalar that is arbitrarily small and $C$ is a constant. The rest of the proof consists of three parts. Firstly, a representation of $\widehat{\beta}^{h}(\theta_{1}) - \widehat{\beta}^{h}(\theta_{2})$ in terms of profiled likelihood is established. Then, the question is transformed to bounding terms related to components of the profiled log likelihood. Lastly, the different pieces are glued together to give an expression of $C$. 

\noindent \textbf{Step 1:} Let $\widehat{Q}(\widehat{\beta}^{h}(\theta); \theta)$ denote the profiled log likelihood function using simulated data $h$,  
\begin{align} 
    \widehat{Q}(\beta; \theta) = \frac{1}{nT}\sum^{n}_{i=1}\sum^{T}_{t=1}&y^{h}_{it}(\theta, \widehat{\alpha}_{i})\ln\Big(\Phi(x_{it}'\beta + \widehat{\gamma}_{i}(\beta)\Big) \nonumber \\
    & + (1-y^{h}_{it}(\theta, \widehat{\alpha}_{i}))\ln\Big(1-\Phi(x_{it}'\beta+ \widehat{\gamma}_{i}(\beta))\Big).
\end{align}
Then by definition, $\widehat{\beta}^{h}(\theta_{1})$ and $\widehat{\beta}^{h}(\theta_{2})$ satisfy the first--order conditions,
$$\frac{\partial \widehat{Q}(\widehat{\beta}^{h}(\theta_{1}); \theta_{1})}{\partial\beta}=0, \quad \frac{\partial \widehat{Q}(\widehat{\beta}^{h}(\theta_{2}); \theta_{2})}{\partial\beta}=0.$$
A first--order Taylor expansion yields
$$\frac{\partial \widehat{Q}(\widehat{\beta}^{h}(\theta_{1}); \theta_{1})}{\partial\beta}=0=\frac{\partial \widehat{Q}(\widehat{\beta}^{h}(\theta_{2}); \theta_{1})}{\partial\beta}+\frac{\partial^{2}\widehat{Q}(\widetilde{\beta}; \theta_{1})}{\partial\beta\partial\beta'}(\widehat{\beta}^{h}(\theta_{1})-\widehat{\beta}^{h}(\theta_{2})),$$
where $\widetilde{\beta}$ is between $\widehat{\beta}^{h}(\theta_{1})$ and $\widehat{\beta}^{h}(\theta_{2})$. Therefore, 
$$\frac{\partial^{2}\widehat{Q}(\widetilde{\beta}; \theta_{1})}{\partial\beta\partial\beta'}(\widehat{\beta}^{h}(\theta_{1})-\widehat{\beta}^{h}(\theta_{2}))=\frac{\partial \widehat{Q}(\widehat{\beta}^{h}(\theta_{2}); \theta_{2})}{\partial\beta} - \frac{\partial \widehat{Q}(\widehat{\beta}^{h}(\theta_{2}); \theta_{1})}{\partial\beta}.$$
Let $\lambda_{s}$ denote the smallest eigenvalue of the Hessian of the profiled likelihood, then a quadratic inequality leads to 
$$\lambda_{s}\lVert\widehat{\beta}^{h}(\theta_{1})-\widehat{\beta}^{h}(\theta_{2})\rVert\leq \Big\lvert \frac{\partial \widehat{Q}(\widehat{\beta}^{h}(\theta_{2}); \theta_{2})}{\partial\beta} - \frac{\partial \widehat{Q}(\widehat{\beta}^{h}(\theta_{2}); \theta_{1})}{\partial\beta}\Big\rvert,$$
where $\frac{\partial\widehat{\gamma}_{i}(\widehat{\beta}(\theta_{2}))}{\partial\beta}=0$ by the envelope theorem. For binary response panel probit models, some algebra leads to the following expression of the right--hand--side term in the absolute sign,
\begin{align} 
    \frac{1}{nT}\sum^{n}_{i=1}\sum^{T}_{t=1}&\Big(y^{h}_{it}(\theta_{1}, \widehat{\alpha}_{i})-y^{h}_{it}(\theta_{2}, \widehat{\alpha}_{i})\Big)\nonumber \\
    & \times \Big(\frac{\phi(x'_{it}\widehat{\beta}^{h}(\theta_{2})+\widehat{\gamma}_{i}(\widehat{\beta}^{h}(\theta_{2}))x_{it}}{\Phi(x'_{it}\widehat{\beta}^{h}(\theta_{2})+\widehat{\gamma}_{i}(\widehat{\beta}^{h}(\theta_{2}))[1-\Phi(x'_{it}\widehat{\beta}^{h}(\theta_{2})+\widehat{\gamma}_{i}(\widehat{\beta}^{h}(\theta_{2}))]}\Big), \label{eqn: original RHS}
\end{align}
where $y^{h}_{it}(\theta)=\boldsymbol{1}\{x'_{it}\theta+\widehat{\alpha}_{i}\geq u^{h}_{it}\}$ and $u^{h}_{it}$ is from the standard normal distribution. Therefore, to establish Condition~(\ref{eqn for stoc eq}), it suffices to focus on Eq.~(\ref{eqn: original RHS}).

\noindent \textbf{Step 2:} By the Cauchy--Schwarz inequality, 
\begin{footnotesize}
\begin{align*}
    \mathbb{E}\Big(\sup_{\lVert\theta_{1}-\theta_{2}\rVert\leq\delta}&\Big\lvert \frac{1}{nT}\sum^{n}_{i=1}\sum^{T}_{t=1}\Big(y^{h}_{it}(\theta_{1}, \widehat{\alpha}_{i})-y^{h}_{it}(\theta_{2}, \widehat{\alpha}_{i})\Big)\\
    & \times \Big(\frac{\phi(x'_{it}\widehat{\beta}^{h}(\theta_{2})+\widehat{\gamma}_{i}(\widehat{\beta}^{h}(\theta_{2}))x_{it}}{\Phi(x'_{it}\widehat{\beta}^{h}(\theta_{2})+\widehat{\gamma}_{i}(\widehat{\beta}^{h}(\theta_{2}))[1-\Phi(x'_{it}\widehat{\beta}^{h}(\theta_{2})+\widehat{\gamma}_{i}(\widehat{\beta}^{h}(\theta_{2}))]}\Big)\Big\rvert\Big) \\
    & \leq \sqrt{\mathbb{E}\Big(\sup_{\lVert\theta_{1}-\theta_{2}\rVert\leq\delta}\Big\lvert \frac{1}{nT}\sum^{n}_{i=1}\sum^{T}_{t=1}y^{h}_{it}(\theta_{1}, \widehat{\alpha}_{i})-y^{h}_{it}(\theta_{2}, \widehat{\alpha}_{i})\Big\rvert^{2}\Big)} \times\\
    & \sqrt{\mathbb{E}\Big(\Big\lvert\frac{1}{nT}\sum^{n}_{i=1}\sum^{T}_{t=1}\frac{\phi(x'_{it}\widehat{\beta}^{h}(\theta_{2})+\widehat{\gamma}_{i}(\widehat{\beta}^{h}(\theta_{2}))x_{it}}{\Phi(x'_{it}\widehat{\beta}^{h}(\theta_{2})+\widehat{\gamma}_{i}(\widehat{\beta}^{h}(\theta_{2}))[1-\Phi(x'_{it}\widehat{\beta}^{h}(\theta_{2})+\widehat{\gamma}_{i}(\widehat{\beta}^{h}(\theta_{2}))]}\Big\lvert^{2}\Big)}.
\end{align*}
\end{footnotesize}
For each $i$ and $t$, the following two $L^{2}$--smoothness conditions hold:
\begin{align}
    & \sqrt{\mathbb{E}\Big(\sup_{\lVert\theta_{1}-\theta_{2}\rVert\leq\delta}\lvert y^{h}_{it}(\theta_{1}, \widehat{\alpha}_{i})-y^{h}_{it}(\theta_{2}, \widehat{\alpha}_{i})\rvert^{2}\Big)}\leq \sqrt{\frac{\mathbb{E}\lVert x_{it}\rVert_{2}}{\sqrt{2\pi}}}\sqrt{\delta}, \label{eqn: l2 for indicators}\\
    & \sqrt{\mathbb{E}\Big(\Big\lvert\frac{\phi(x'_{it}\widehat{\beta}^{h}(\theta_{2})+\widehat{\gamma}_{i}(\widehat{\beta}^{h}(\theta_{2}))x_{it}}{\Phi(x'_{it}\widehat{\beta}^{h}(\theta_{2})+\widehat{\gamma}_{i}(\widehat{\beta}^{h}(\theta_{2}))[1-\Phi(x'_{it}\widehat{\beta}^{h}(\theta_{2})+\widehat{\gamma}_{i}(\widehat{\beta}^{h}(\theta_{2}))]}\Big\lvert^{2}\Big)} \leq K_{2}, \label{eqn: l2 for second term}
\end{align}
where $\lVert x\rVert_{2}$ denotes the $L_{2}$-norm $\lvert x'x \rvert^{1/2}$. This corresponds to type IV class in \cite{Andrews1994}.

\noindent \textbf{Proving condition~(\ref{eqn: l2 for indicators}):} Denote $\Delta\theta:=\theta_{2} - \theta_{1}$ and note that
\begin{align*}
    \sup_{\lVert\theta_{1}-\theta_{2}\rVert\leq\delta}\lvert y^{h}_{it}(\theta_{1})-y^{h}_{it}(\theta_{2})\lvert & = \sup_{\lVert\Delta\theta\rVert\leq\delta}\lvert\boldsymbol{1}\{x'_{it}\theta_{1}+\widehat{\alpha}_{i}\geq u^{h}_{it}\}-\boldsymbol{1}\{x'_{it}(\theta_{1}+\Delta\theta)+\widehat{\alpha}_{i}\geq u^{h}_{it}\}\lvert.
\end{align*}
The direction that obtains the supremum is given by
$$\Delta\theta = \pm \frac{\delta}{\lVert x_{it}\rVert_{2}}x_{it}.$$
Therefore
\begin{equation} \label{eqn: max bound of indicators}
     \mathbb{E}\Big[\sup_{\lVert\theta_{1}-\theta_{2}\rVert\leq\delta}\lvert y^{h}_{it}(\theta_{1})-y^{h}_{it}(\theta_{2})\lvert\Big] \leq \mathbb{E}\Big(\boldsymbol{1}\{x'_{it}\theta_{1}+\widehat{\alpha}_{i}\geq u^{h}_{it}\}-\boldsymbol{1}\{x'_{it}\theta_{1}-\lVert x_{it}\rVert_{2}\delta+\widehat{\alpha}_{i}\geq u^{h}_{it}\}\Big).
\end{equation}
Because $\delta$ is a scalar, a proof strategy {\`a} la  \cite{ChenLintonVanKeilegom2003} is employed to bound the right--hand--side term in Equation~(\ref{eqn: max bound of indicators}). More specifically, note that $$\boldsymbol{1}\{x'_{it}\theta_{1}+\widehat{\alpha}_{i}\geq u^{h}_{it}\}-\boldsymbol{1}\{x'_{it}\theta_{1}-\lVert x_{it}\rVert_{2}\delta+\widehat{\alpha}_{i}\geq u^{h}_{it}\}$$
takes value either 1 or 0, and the expectation is the probability that the following event occurs: 
$$x'_{it}\theta_{1}+\widehat{\alpha}_{i}\geq u^{h}_{it} \geq x'_{it}\theta_{1}-\lVert x_{it}\rVert_{2}\delta+\widehat{\alpha}_{i}.$$
Applying law of iterated expectation on the right--hand--side term and first--order Taylor expansion around $\delta$,
\begin{align*}
    &\mathbb{E}\Big[\mathbb{E}\Big(\boldsymbol{1}\{x'_{it}\theta_{1}+\widehat{\alpha}_{i}\geq u^{h}_{it}\}-\boldsymbol{1}\{x'_{it}\theta_{1}-\lVert x_{it}\rVert_{2}\delta+\widehat{\alpha}_{i}\geq u^{h}_{it}\} \mid x'_{it}, \widehat{\alpha}_{i}\Big)\Big] \\
    =& \mathbb{E}\Big[\Phi(x'_{it}\theta_{1}+\widehat{\alpha}_{i}) - \Phi(x'_{it}\theta_{1} -\lVert x_{it}\rVert_{2}\delta +\widehat{\alpha}_{i})\Big] \\ 
    =& \mathbb{E}\Big[\phi(x'_{it}\theta + \widehat{\alpha}_{i})\lVert x_{it}\rVert_{2}\Big]\delta
\end{align*}
Therefore,
$$\mathbb{E}\Big[\sup_{\lVert\theta_{1}-\theta_{2}\rVert\leq\delta}\lvert y^{h}_{it}(\theta_{1})-y^{h}_{it}(\theta_{2})\lvert\Big] \leq \mathbb{E}\Big[\phi(x'_{it}\theta + \widehat{\alpha}_{i})\lVert x_{it}\rVert_{2}\Big]\delta \leq \frac{\mathbb{E}\lVert x_{it}\rVert_{2}}{\sqrt{2\pi}}\delta,$$
where the last inequality uses the fact that $\phi(\cdot)\leq\frac{1}{\sqrt{2\pi}}$. 

\noindent \textbf{Proving condition~(\ref{eqn: l2 for second term}):} Note that
\begin{align*}
    &\sqrt{\mathbb{E}\Big(\Big\lvert\frac{\phi(x'_{it}\widehat{\beta}^{h}(\theta_{2})+\widehat{\gamma}_{i}(\widehat{\beta}^{h}(\theta_{2}))x_{it}}{\Phi(x'_{it}\widehat{\beta}^{h}(\theta_{2})+\widehat{\gamma}_{i}(\widehat{\beta}^{h}(\theta_{2}))[1-\Phi(x'_{it}\widehat{\beta}^{h}(\theta_{2})+\widehat{\gamma}_{i}(\widehat{\beta}^{h}(\theta_{2}))]}\Big\lvert^{2}\Big)}
\end{align*}
is no greater than 
\begin{align*}
    \sqrt{\mathbb{E}\Big(\sup_{(\beta, \gamma)\in\mathcal{B}\times\Gamma_{\gamma}}\Big\lvert\frac{\phi(x'_{it}\beta+\gamma)x_{it}}{\Phi(x'_{it}\beta+\gamma)[1-\Phi(x'_{it}\beta+\gamma)]}\Big\lvert^{2}\Big)},
\end{align*}
which is bounded based on Lipschitz condition. 

\noindent \textbf{Step 3:} Because the supremum of sum is no greater than sum of the supremum, 
\begin{align*}
    &\mathbb{E}\Big(\sup_{\lVert\theta_{1}-\theta_{2}\rVert\leq\delta}\Big\lvert \frac{1}{nT}\sum^{n}_{i=1}\sum^{T}_{t=1} y^{h}_{it}(\theta_{1}, \widehat{\alpha}_{i})-y^{h}_{it}(\theta_{2}, \widehat{\alpha}_{i})\Big\rvert^{2}\Big)\\
    &\leq \frac{1}{nT}\sum^{n}_{i=1}\sum^{T}_{t=1}\mathbb{E}\Big(\sup_{\lVert\theta_{1}-\theta_{2}\rVert\leq\delta}\lvert y^{h}_{it}(\theta_{1}, \widehat{\alpha}_{i})-y^{h}_{it}(\theta_{2}, \widehat{\alpha}_{i})\rvert^{2}\Big) \\
    & \leq \frac{\delta}{\sqrt{2\pi}}\frac{1}{nT}\sum^{n}_{i=1}\sum^{T}_{t=1}\mathbb{E}\lVert x_{it}\rVert_{2}, 
\end{align*}
and
\begin{align*}
    &\mathbb{E}\Big(\Big\lvert\frac{1}{nT}\sum^{n}_{i=1}\sum^{T}_{t=1}\frac{\phi(x'_{it}\widehat{\beta}(\theta_{2})+\widehat{\gamma}_{i}(\widehat{\beta}(\theta_{2}))x_{it}}{\Phi(x'_{it}\widehat{\beta}(\theta_{2})+\widehat{\gamma}_{i}(\widehat{\beta}(\theta_{2}))[1-\Phi(x'_{it}\widehat{\beta}(\theta_{2})+\widehat{\gamma}_{i}(\widehat{\beta}(\theta_{2}))]}\Big\lvert^{2}\Big) \\
    & \leq \frac{1}{nT}\sum^{n}_{i=1}\sum^{T}_{t=1}\mathbb{E}K_{it}.
\end{align*}
Therefore,
\begin{align*}
     \mathbb{E}\Big(\sup_{\lVert \theta_{1}-\theta_{2}\rVert\leq \delta}\lVert \widehat{\beta}(\theta_{1})-\widehat{\beta}(\theta_{2})\rVert\Big)\leq \frac{\sqrt{\delta}}{(2\pi)^{1/4}}\sqrt{\frac{1}{nT}\sum^{n}_{i=1}\sum^{T}_{t=1}\mathbb{E}\lVert x_{it}\rVert_{2}K_{it}}.
\end{align*}
This verifies condition~(\ref{eqn for stoc eq}) and hence establishes the stochastic equicontinuity condition. 

\noindent \textbf{Step 4:} By Theorem 1 in \cite{Newey1991}, $\widehat{\beta}^{h}(\theta, \boldsymbol{\widehat{\alpha}})$ converges to $\beta(\theta, \boldsymbol{\alpha_{0}})$ uniformly over $\theta\in\Theta$. 
\end{proof}
\subsection{Proof of Theorem~\ref{consistency}}
\begin{proof}
Following  the argument as in Appendix 1 of \cite{GourierouxMonfortRenault1993}, consistency of $\widetilde{\theta}^{H}$ requires the following three conditions to hold: 
\begin{enumerate}
    \item the function $\beta(\theta, \boldsymbol{\alpha_{0}})$ is invertible;
    \item $\widehat{\theta}$ converges to $\beta(\theta_{0}, \boldsymbol{\alpha_{0}})$ in $\theta_{0}\in\Theta$  pointwise;
    \item $\widehat{\beta}^{h}(\theta, \boldsymbol{\widehat{\alpha}})$ converges to $\beta(\theta, \boldsymbol{\alpha_{0}})$ uniformly over $\theta\in\Theta$.
\end{enumerate}  
The first condition is satisfied because function is an identity. The second condition only involves fixed effect estimator using observed data, and is a standard result in large--$T$ panel literature \citep[e.g,][Theorem 3]{HahnKuersteiner2011}. The third condition is verified by Proposition~\ref{Proposition: uniform consistency}. 

\noindent Recall that $\widetilde{\theta}^{H}$ is the solution to the optimization problem: 
\begin{align*}
    \widetilde{\theta}^{H}&=\argmin_{\theta\in\Theta} [\widehat{\theta}-\widehat{\beta}_{H}(\theta,\boldsymbol{\widehat{\alpha}})]'[\widehat{\theta}-\widehat{\beta}_{H}(\theta,\boldsymbol{\widehat{\alpha}})],
\end{align*}
where $\widehat{\beta}_{H}(\theta,\boldsymbol{\widehat{\alpha}} ):=\frac{1}{H}\sum^{H}_{h=1}\widehat{\beta}^{h}(\theta, \boldsymbol{\widehat{\alpha}})$. Therefore. the limit of the optimization problem becomes
$$\min_{\theta\in\Theta}[\theta_{0}-\theta]'[\theta_{0}-\theta],$$
which has a unique solution $\theta_{0}$. Therefore, 
$$\widetilde{\theta}^{H}\xrightarrow{p}\theta_{0}.$$
\end{proof}
\subsection{Proof of Theorem~\ref{CLT}} \label{Appendix: CLT proof}
\begin{proof}
By Assumption~\ref{asu: replace smooth} and consistency of $\widetilde{\theta}^{H}$,
$$\widehat{\theta}=\widehat{\beta}_{H}(\widetilde{\theta}^{H}, \boldsymbol{\widehat{\alpha}})=\widehat{\beta}_{H}(\theta_{0}, \boldsymbol{\widehat{\alpha}})+\mathbb{E}(\widehat{\beta}_{H}(\widetilde{\theta}^{H}, \boldsymbol{\widehat{\alpha}}) - \widehat{\beta}_{H}(\theta_{0},\boldsymbol{\widehat{\alpha}}))+o_{p}\Big(\frac{1}{\sqrt{nT}}\Big).$$
By the mean--value theorem, 
$$\mathbb{E}(\widehat{\beta}_{H}(\widetilde{\theta}^{H},\boldsymbol{\widehat{\alpha}}) - \widehat{\beta}_{H}(\theta_{0},\boldsymbol{\widehat{\alpha}}))=\frac{\partial\mathbb{E}\widehat{\beta}_{H}(\overline{\theta},\boldsymbol{\widehat{\alpha}})}{\partial\theta}(\widetilde{\theta}^{H}-\theta_{0}),$$
where $\overline{\theta}$ is between $\theta_{0}$ and $\widetilde{\theta}^{H}$. Therefore,
\begin{align*}
    \sqrt{nT}(\widetilde{\theta}^{H}-\theta_{0})& = -\Big[\frac{\partial\mathbb{E}\widehat{\beta}_{H}(\overline{\theta},\boldsymbol{\widehat{\alpha}})}{\partial\theta}\Big]^{-1}\sqrt{nT}\Big(\widehat{\beta}_{H}(\theta_{0},\boldsymbol{\widehat{\alpha}})-\widehat{\theta}\Big) \\
    & = \sqrt{nT}\Big(\widehat{\theta}-\widehat{\beta}_{H}(\theta_{0},\boldsymbol{\widehat{\alpha}})\Big)+o_{p}(1),
\end{align*}
where the last equality uses the property that $\beta(\theta, \boldsymbol{\alpha_{0}})=\theta$. Therefore, it suffices to focus on $\sqrt{nT}(\widehat{\theta}- \widehat{\beta}_{H}(\theta_{0},\boldsymbol{\widehat{\alpha}}))$. \cite{HahnKuersteiner2011} derive the representation of $\widehat{\theta}-\theta_{0}$ as follows:
\begin{equation*} 
    \widehat{\theta}-\theta_{0}=\frac{A(\theta_{0}, \boldsymbol{\alpha_{0}})}{\sqrt{nT}}+\frac{B(\theta_{0}, \boldsymbol{\alpha_{0}})}{T}+o_{p}\Big(\frac{1}{T}\Big),
\end{equation*}
where $A(\theta_{0}, \boldsymbol{\alpha_{0}})$ and $B(\theta_{0}, \boldsymbol{\alpha_{0}})$ are complicated functions of the high--order derivatives of the log likelihood. Because the same regression is run on simulated data, 
\begin{equation*} 
    \widehat{\beta}_{H}(\theta_{0}) - \beta(\theta_{0}, \boldsymbol{\alpha_{0}}) =  \frac{A^{h}(\theta_{0}, \boldsymbol{\widehat{\alpha}})}{\sqrt{nT}}+\frac{B^{h}(\theta_{0}, \boldsymbol{\widehat{\alpha}})}{T}+o_{p}\Big(\frac{1}{T}\Big),
\end{equation*} 
where $\boldsymbol{\widehat{\alpha}}:=(\widehat{\alpha}_{1},\dots,\widehat{\alpha}_{n})$. This implies
$$\widehat{\beta}_{H}(\theta_{0}, \boldsymbol{\widehat{\alpha}})=\beta(\theta_{0}, \boldsymbol{\alpha_{0}}) + \frac{1}{H}\sum^{H}_{h=1}\frac{A^{h}(\theta_{0}, \boldsymbol{\widehat{\alpha}})}{\sqrt{nT}}+\frac{1}{H}\sum^{H}_{h=1}\frac{B^{h}(\theta_{0}, \boldsymbol{\widehat{\alpha}})}{T}+o_{p}\Big(\frac{1}{T}\Big).$$
Combined with $\beta(\theta_{0}, \boldsymbol{\alpha_{0}})=\theta_{0}$,
\begin{align*}
    \sqrt{nT}\Big(\widehat{\theta} - \widehat{\beta}_{H}(\theta_{0}, \boldsymbol{\widehat{\alpha}})\Big) = \Big(&A(\theta_{0}, \boldsymbol{\alpha_{0}})-\frac{1}{H}\sum^{H}_{h=1}A^{h}(\theta_{0}, \boldsymbol{\widehat{\alpha}})\Big) \\ &+\sqrt{\frac{n}{T}}\Big(B(\theta_{0}, \boldsymbol{\alpha_{0}})-\frac{1}{H}\sum^{H}_{h=1}B^{h}(\theta_{0}, \boldsymbol{\widehat{\alpha}})\Big)+o_{p}\Big(\sqrt{\frac{n}{T}}\Big).
\end{align*}
The rest of the proof shows that bias term cancels out and the asymptotic normality holds. To simplify notation, the rest of the proof proceeds by setting $H=1$.

\noindent\textbf{Step 1:} Bias correction is established in Appendix~\ref{proof of prop 2}. 

\noindent \textbf{Step 2:} The simulation analog of the CLT term is 
$$A^{h}(\theta_{0}, \boldsymbol{\widehat{\alpha}})=\Big(\frac{1}{n}\sum^{n}_{i=1}\mathcal{I}_{i}(\theta_{0}, \widehat{\alpha}_{i})\Big)^{-1}\frac{1}{\sqrt{n}}\sum^{n}_{i=1}\frac{1}{\sqrt{T}}\sum^{T}_{t=1}U^{h}_{it}(\theta, \widehat{\alpha}_{i})$$
Note that   
\begin{align*}
    & \Big(\frac{1}{n}\sum^{n}_{i=1}\mathcal{I}_{i}(\theta_{0}, \widehat{\alpha}_{i})\Big)^{-1}-\Big(\frac{1}{n}\sum^{n}_{i=1}\mathcal{I}_{i}(\theta_{0}, \alpha_{i0})\Big)^{-1} \\
    = & \Big(\frac{1}{n}\sum^{n}_{i=1}\mathcal{I}_{i}(\theta_{0}, \widehat{\alpha}_{i})\Big)^{-1}\Big(\frac{1}{n}\sum^{n}_{i=1}\big[\mathcal{I}_{i}(\theta_{0}, \alpha_{i0})-\mathcal{I}_{i}(\theta_{0}, \widehat{\alpha}_{i})\big]\Big)\Big(\frac{1}{n}\sum^{n}_{i=1}\mathcal{I}_{i}(\theta_{0}, \alpha_{i0})\Big)^{-1}.
\end{align*}
By continuous mapping theorem, $\mathcal{I}_{i}(\theta_{0}, \widehat{\alpha}_{i})\xrightarrow{p}\mathcal{I}_{i}(\theta_{0}, \alpha_{i0})$ for each $i$, and thus 
$$\Big\lvert\Big(\frac{1}{n}\sum^{n}_{i=1}\mathcal{I}_{i}(\theta_{0}, \widehat{\alpha}_{i})\Big)^{-1} - \Big(\frac{1}{n}\sum^{n}_{i=1}\mathcal{I}_{i}(\theta_{0}, \alpha_{i0})\Big)^{-1}\Big\rvert\xrightarrow{p} 0.$$
Combined with Assumption~\ref{assumption for normality}, $A^{h}(\theta_{0}, \boldsymbol{\widehat{\alpha}})$ has the same distribution as $A^{h}(\theta_{0}, \boldsymbol{\alpha_{0}})$, and by Proposition 5 in \cite{GourierouxMonfortRenault1993}, 
$$A(\theta_{0}, \boldsymbol{\alpha_{0}})-\frac{1}{H}\sum^{H}_{h=1}A^{h}(\theta_{0}, \boldsymbol{\alpha_{0}})\sim\mathcal{N}(0, (1+\frac{1}{H})\mathbb{E}(A(\theta_{0}, \boldsymbol{\alpha_{0}})A(\theta_{0}, \boldsymbol{\alpha_{0}})')).$$
\end{proof}
\subsection{Proof of Proposition~\ref{prop: bias correction}} \label{proof of prop 2}
\begin{proof}
Consider an infeasible fixed effect estimator $\widehat{\beta}_{H}(\theta_{0}, \alpha_{0})$ that is obtained from data simulated by $(\theta_{0}, \alpha_{0})$. Then the representation of $\widehat{\beta}_{H}(\theta_{0}, \alpha_{0}) - \theta_{0}$ takes the form
\begin{equation*} 
    \widehat{\beta}_{H}(\theta_{0}, \boldsymbol{\alpha_{0}}) - \theta_{0} =  \frac{A^{h}(\theta_{0}, \boldsymbol{\alpha_{0}})}{\sqrt{nT}}+\frac{B^{h}(\theta_{0}, \boldsymbol{\alpha_{0}})}{T}+o\Big(\frac{1}{T}\Big),
\end{equation*} 
where the superscript $h$ denotes the fact that the dependent variable in $B^{h}(\theta_{0}, \boldsymbol{\alpha_{0}})$ is $y^{h}_{it}(\theta_{0}, \alpha_{i0})$. 
Because $B(\theta_{0}, \boldsymbol{\alpha_{0}})$ and $B^{h}(\theta_{0}, \boldsymbol{\alpha_{0}})$ have the same probability limit, they converge to the same expectation, which is the asymptotic bias. Therefore, it suffices to show that $B^{h}(\theta_{0}, \boldsymbol{\widehat{\alpha}})$ uniformly well approximates $B^{h}(\theta_{0}, \boldsymbol{\alpha_{0}}).$

Now prove bias correction of the following form:
$$\lvert B^{h}(\theta, \boldsymbol{\widehat{\alpha}})  - B(\theta, \boldsymbol{\alpha_{0}}) \lvert\xrightarrow{p} 0.$$
By Markov inequality, $\forall \eta>0$, 
$$Pr(\lvert B^{h}(\theta, \boldsymbol{\widehat{\alpha}})  - B(\theta, \boldsymbol{\alpha_{0}}) \lvert \geq \eta)\leq \frac{1}{\eta}\mathbb{E}(\lvert B^{h}(\theta, \boldsymbol{\widehat{\alpha}})  - B(\theta, \boldsymbol{\alpha_{0}}) \lvert).$$
Therefore it suffices to bound the RHS term. 
By the triangular inequality, 
\begin{align} \label{eqn: main equation for bias correction}
    &\mathbb{E}(\lvert B^{h}(\theta, \boldsymbol{\widehat{\alpha}}) - B(\theta, \boldsymbol{\alpha_{0}})\lvert) \nonumber \\
    &\leq \mathbb{E}(\lvert B^{h}(\theta, \boldsymbol{\widehat{\alpha}}) - B^{h}(\theta, \boldsymbol{\alpha_{0}})\lvert) + \mathbb{E}(\lvert B^{h}(\theta, \boldsymbol{\alpha_{0}}) - B(\theta, \boldsymbol{\alpha_{0}})\lvert).
\end{align}
The second RHS term in equation~(\ref{eqn: main equation for bias correction}) is $o_{p}(1)$ because $B^{h}(\theta,\boldsymbol{\alpha_{0}})$ and $B(\theta, \boldsymbol{\alpha_{0}})$ have the same probability limit. Regarding the first RHS term, by the  triangular inequality, 
\begin{align}
    \mathbb{E}\lvert B^{h}(\theta, \boldsymbol{\widehat{\alpha}}) - B^{h}(\theta, \boldsymbol{\alpha_{0}})\lvert 
    & \leq \mathbb{E}\Big\lvert\Big(\frac{1}{n}\sum^{n}_{i=1}\mathcal{I}_{i}(\theta_{0}, \widehat{\alpha}_{i})\Big)^{-1}\frac{1}{n}\sum^{n}_{i=1}B^{h}_{i}(\theta_{0}, \widehat{\alpha}_{i}) \nonumber\\
    & - \Big(\frac{1}{n}\sum^{n}_{i=1}\mathcal{I}_{i}(\theta_{0}, \widehat{\alpha}_{i})\Big)^{-1}\frac{1}{n}\sum^{n}_{i=1}B^{h}_{i}(\theta_{0}, \alpha_{i0}) \nonumber\\
    & + \Big(\frac{1}{n}\sum^{n}_{i=1}\mathcal{I}_{i}(\theta_{0}, \widehat{\alpha}_{i})\Big)^{-1}\frac{1}{n}\sum^{n}_{i=1}B^{h}_{i}(\theta_{0}, \alpha_{i0}) \nonumber \\
    & - \Big(\frac{1}{n}\sum^{n}_{i=1}\mathcal{I}_{i}(\theta_{0}, \alpha_{i0})\Big)^{-1}\frac{1}{n}\sum^{n}_{i=1}B^{h}_{i}(\theta_{0}, \alpha_{i0})\Big\rvert \nonumber\\
    & \leq \mathbb{E}\Big\lvert\Big(\frac{1}{n}\sum^{n}_{i=1}\mathcal{I}_{i}(\theta_{0}, \widehat{\alpha}_{i})\Big)^{-1}\Big\rvert\nonumber\\
    &\times\Big\lvert\frac{1}{n}\sum^{n}_{i=1}\big[B^{h}_{i}(\theta_{0}, \widehat{\alpha}_{i}) - B^{h}_{i}(\theta_{0}, \alpha_{i0})\big]\Big\rvert \label{eqn: 1st term bias}\\
    & + \mathbb{E}\Big\lvert\Big(\frac{1}{n}\sum^{n}_{i=1}\mathcal{I}_{i}(\theta_{0}, \widehat{\alpha}_{i})\Big)^{-1}-\Big(\frac{1}{n}\sum^{n}_{i=1}\mathcal{I}_{i}(\theta_{0}, \alpha_{i0})\Big)^{-1}\Big\rvert \times \nonumber \\ &\Big\lvert\frac{1}{n}\sum^{n}_{i=1}B^{h}_{i}(\theta_{0}, \alpha_{i0})\Big\rvert \label{eqn: 2nd term bias}.
\end{align}
Therefore, it suffices to focus on bounding terms~(\ref{eqn: 1st term bias}) and (\ref{eqn: 2nd term bias}). 

\noindent For term~(\ref{eqn: 2nd term bias}), note that 
\begin{align*}
    & \Big(\frac{1}{n}\sum^{n}_{i=1}\mathcal{I}_{i}(\theta_{0}, \widehat{\alpha}_{i})\Big)^{-1}-\Big(\frac{1}{n}\sum^{n}_{i=1}\mathcal{I}_{i}(\theta_{0}, \alpha_{i0})\Big)^{-1} \\
    = & \Big(\frac{1}{n}\sum^{n}_{i=1}\mathcal{I}_{i}(\theta_{0}, \widehat{\alpha}_{i})\Big)^{-1}\Big(\frac{1}{n}\sum^{n}_{i=1}\big[\mathcal{I}_{i}(\theta_{0}, \alpha_{i0})-\mathcal{I}_{i}(\theta_{0}, \widehat{\alpha}_{i})\big]\Big)\Big(\frac{1}{n}\sum^{n}_{i=1}\mathcal{I}_{i}(\theta_{0}, \alpha_{i0})\Big)^{-1}.
\end{align*}
By continuous mapping theorem, $\mathcal{I}_{i}(\theta_{0}, \widehat{\alpha}_{i})\xrightarrow{p}\mathcal{I}_{i}(\theta_{0}, \alpha_{i0})$ for each $i$. Therefore, 
\begin{align*}
    &\mathbb{E}\Big\lvert\Big(\frac{1}{n}\sum^{n}_{i=1}\mathcal{I}_{i}(\theta_{0}, \widehat{\alpha}_{i})\Big)^{-1}-\Big(\frac{1}{n}\sum^{n}_{i=1}\mathcal{I}_{i}(\theta_{0}, \alpha_{i0})\Big)^{-1}\Big\rvert \cdot \Big\lvert\frac{1}{n}\sum^{n}_{i=1}B^{h}_{i}(\theta_{0}, \alpha_{i0})\Big\rvert \xrightarrow{p} 0. 
\end{align*}
For term~(\ref{eqn: 1st term bias}), note that 
\begin{align*}
    \Big\lvert\frac{1}{n}\sum^{n}_{i=1}\big[B^{h}_{i}(\theta_{0}, \widehat{\alpha}_{i}) - B^{h}_{i}(\theta_{0}, \alpha_{i0})\big]\Big\rvert & \leq \frac{1}{n}\sum^{n}_{i=1}\lvert B^{h}_{i}(\theta_{0}, \widehat{\alpha}_{i}) - B^{h}_{i}(\theta_{0}, \alpha_{i0}) \rvert \\
    & \leq \max_{1\leq i\leq n} \lvert B^{h}_{i}(\theta_{0}, \widehat{\alpha}_{i}) - B^{h}_{i}(\theta_{0}, \alpha_{i0}) \rvert.
\end{align*}
Therefore,
\begin{align*}
    & \mathbb{E}\Big\lvert\Big(\frac{1}{n}\sum^{n}_{i=1}\mathcal{I}_{i}(\theta_{0}, \widehat{\alpha}_{i})\Big)^{-1}\Big\rvert\cdot\Big\lvert\frac{1}{n}\sum^{n}_{i=1}\big[B^{h}_{i}(\theta_{0}, \widehat{\alpha}_{i}) - B^{h}_{i}(\theta_{0}, \alpha_{i0})\big]\Big\rvert \\
    \leq & \mathbb{E}\Big\lvert\Big(\frac{1}{n}\sum^{n}_{i=1}\mathcal{I}_{i}(\theta_{0}, \widehat{\alpha}_{i})\Big)^{-1}\Big\rvert\cdot\max_{1\leq i\leq n} \lvert B^{h}_{i}(\theta_{0}, \widehat{\alpha}_{i}) - B^{h}_{i}(\theta_{0}, \alpha_{i0}) \rvert \\
    \leq & \sqrt{\mathbb{E}\Big\lvert\Big(\frac{1}{n}\sum^{n}_{i=1}\mathcal{I}_{i}(\theta_{0}, \widehat{\alpha}_{i})\Big)^{-1}\Big\rvert^{2}}\cdot \sqrt{\mathbb{E}\max_{1\leq i\leq n}\Big\lvert B^{h}_{i}(\theta_{0}, \widehat{\alpha}_{i}) - B^{h}_{i}(\theta_{0}, \alpha_{i0})\Big\rvert^{2}},
\end{align*}
where the second inequality is due to Cauchy--Schwarz inequality. By continuous mapping theorem, $\mathcal{I}_{i}(\theta_{0}, \widehat{\alpha}_{i})\xrightarrow{p}\mathcal{I}_{i}(\theta_{0}, \alpha_{i0})$, and by Slutsky theorem, 
$$\sqrt{\mathbb{E}\Big\lvert\Big(\frac{1}{n}\sum^{n}_{i=1}\mathcal{I}_{i}(\theta_{0}, \widehat{\alpha}_{i})\Big)^{-1}\Big\rvert^{2}}\xrightarrow{p}\sqrt{\mathbb{E}\Big\lvert\Big(\frac{1}{n}\sum^{n}_{i=1}\mathcal{I}_{i}(\theta_{0}, \alpha_{i0})\Big)^{-1}\Big\rvert^{2}}.$$
Combined with Assumption~\ref{asu: tech for bias}, 
$$\mathbb{E}\Big\lvert\Big(\frac{1}{n}\sum^{n}_{i=1}\mathcal{I}_{i}(\theta_{0}, \widehat{\alpha}_{i})\Big)^{-1}\Big\rvert\cdot\Big\lvert\frac{1}{n}\sum^{n}_{i=1}\big[B^{h}_{i}(\theta_{0}, \widehat{\alpha}_{i}) - B^{h}_{i}(\theta_{0}, \alpha_{i0})\big]\Big\rvert\xrightarrow{p}0.$$
\end{proof}
\newpage
\section{Computation Appendix} \label{Appendix: Computation}
\subsection{Calibration Procedures} \label{appendix: calibration}
Simulation procedures for the labor force participation application. 
\begin{enumerate}
    \item Run the regression on the LFP data to obtain $\widetilde{\theta}$ and $\widetilde{\alpha}_{i}$'s. These are treated as true coefficients for the calibration exercise. 
    \item For each simulation $s = 1,\dots, S$, create a synthetic panel data based on the equation
    $$y_{it}^{s}=\boldsymbol{1}\{X_{it}\widetilde{\theta}+\widetilde{\alpha}_{i}>u^{s}_{it}\},$$
    where $u^{s}_{it}\sim iid \mathcal{N}(0, 1)$. The data $\{(y^{s}_{it}, X_{it})\}$ are considered as the observed data for simulation $s$. 
    \item \textbf{Implementing the estimation:} 
    \begin{enumerate}
        \item Run Probit regression on $\{(y^{s}_{it}, X_{it})\}$ and obtain $\widehat{\theta}^{s}$ and $\widehat{\alpha}^{s}_{i}$. This denotes the fixed effect estimators using observed data.    
        \item \textbf{Data simulation}: 
        \begin{enumerate}
            \item Given a set of parameter $\theta$, simulate dependent variable using
            $$y_{it}^{h}(\theta)=\boldsymbol{1}\{X'_{it}\theta+\widehat{\alpha}^{s}_{i}>\varepsilon^{h}_{it}\}, \quad \varepsilon^{h}_{it}\sim iid \mathcal{N}(0, 1)$$
            Run Probit regression on $\{y^{h}_{it}(\theta), X_{it}\} $to obtain $\widehat{\beta}^{h}(\theta)$.
            \item Repeat step (i) for $H=10$ times and compute 
            $$\widehat{\beta}^{H}(\theta) = \frac{1}{H}\sum^{H}_{h=1}\widehat{\beta}^{h}(\theta).$$
            \item Compute the indirect inference estimator $\widetilde{\theta}^{H}$ by solving the following equation 
            $$\widehat{\theta}^{s} = \widehat{\beta}^{H}(\widetilde{\theta}^{H}).$$
        \end{enumerate}
    \end{enumerate}
    \item Repeat steps 2 and 3 for $S=500$ times.
\end{enumerate}
\subsection{Simulations for Dynamic Labor Force Participation} \label{appendix: dynamic simulation}
This subsection introduces dynamics into the specification and compare the performance of indirect fixed effect estimators with other estimators. 

Positive serial correlation observed in employment outcomes motivates the question of identifying state dependence, i.e., the causal impact of past employment on future employment for married women. However, the positive correlation can also be driven by individual--specific unobserved heterogeneity such as willingness to work. Therefore, an important question of interest is to distinguish between state dependence and persistent unobserved heterogeneity.

Following the empirical specification in \cite{Fernandez-Val2009}, this paper controls for time--invariant unobserved heterogeneity by adding individual fixed effects, 
\begin{equation} 
    y_{it}=\boldsymbol{1}\{X'_{it}\theta+\alpha_{i}\geq u_{it}\}, \quad u_{it}\sim\mathcal{N}(0, 1),
\end{equation}
where the vector of pre--determined covariates $X_{it}:=(x_{it}, y_{i,t-1})$ now contains an extra variable: $y_{i, t-1}$, which denotes the lagged dependent variable. The first year of the sample is excluded for use as the initial condition in the dynamic model. In the data simulation step, the dependent variable at time $t$ has the following representation:
\begin{equation}
    y^{h}_{it}(\theta, \widehat{\alpha}_{i}) = \boldsymbol{1}\{\theta_{1}y^{h}_{i,t-1}(\theta, \widehat{\alpha}_{i}) + x_{it}'\theta_{-1} + \widehat{\alpha}_{i}\geq u^{h}_{it}\}, \quad u^{h}_{it}\sim\mathcal{N}(0, 1). 
\end{equation}
where $\theta_{-1}$ denotes parameters other than the one for $y^{h}_{i,t-1}$.

Table~(\ref{empirics: dynamics}) reports the coefficients estimates using different methods. The analytical bias correction (ABC) corresponds to the method proposed by \cite{HahnKuersteiner2011} and serves as a benchmark. The JBC method by \cite{HahnNewey2004} is no longer applicable due to dynamics in the specification. The results are similar to the static case. When $H=20$, the indirect inference estimator produces bias correction results close to the ABC. On the other hand, the HBC estimate of lagged LFP is larger. Regarding the standard errors, the indirect fixed effect estimator does not inflate the errors when $H=20$, but HBC has larger standard errors across all variables. 
\begin{table}[htbp!]
\centering
\begin{threeparttable}
\caption{Parameter Estimates for Dynamic LFP} \label{empirics: dynamics}
\begin{tabular}{rrrrrrrr}
  \hline\hline
  & lfp\_lagged & kids0\_2 & kids3\_5 & kids6\_17 & loghusinc & age & age2 \Tstrut\Bstrut\\ 
  \hline
FE  &  0.76 & -0.55  & -0.28 & -0.07 & -0.25 & 2.05 & -0.25  \Tstrut\Bstrut\\ 
    & (0.04) & (0.06) & (0.05) & (0.04) & (0.06) & (0.38) & (0.05)\vspace{0.08cm}\\ 
IFE--1 & 0.80  & -0.41 & -0.25 &  -0.06 & -0.31 & 2.04 & -0.24  \\
    & (0.06) & (0.08) & (0.08) & (0.06) & (0.08) & (0.54) & (0.07)\vspace{0.08cm}\\
IFE--10 & 1.09 & -0.39  & -0.07 & -0.04 & -0.32 &  1.78 & -0.19 \\
    & (0.04) & (0.06) & (0.06) & (0.04) & (0.06) & (0.40) & (0.05)\vspace{0.08cm}\\
IFE--20 & 1.11 & -0.48 & -0.22 & -0.07 & -0.28 & 1.75 & -0.23  \\
    & (0.04) & (0.06) & (0.05) & (0.04) & (0.06) & (0.39) & (0.05)\vspace{0.08cm}\\
HBC & 1.35  & -0.63 & -0.34 & -0.15 & -0.31 & 1.79  &   -0.20 \\
 & (0.05) & (0.09) & (0.09) & (0.08) & (0.07) & (0.88) & (0.12)\vspace{0.08cm}\\
ABC & 0.99 & -0.48 & -0.21 & -0.06 & -0.23 & 1.84 & -0.22 \\
 & (0.04) & (0.06) & (0.05) & (0.04) & (0.06) & (0.38) & (0.05)\vspace{0.08cm}\\
\hline
\bottomrule
\end{tabular}
\begin{tablenotes}[flushleft]
\linespread{1}\footnotesize
\item\hspace*{-\fontdimen2\font}\textit{Notes:} Standard errors are stored in the parenthesis and are computed based on the Hessian matrix of profiled log likelihood. For details of the HBC estimates and standard errors computation, refer to page 1025 in \cite{DhaeneJochmans2015}. 
\end{tablenotes}
\end{threeparttable}   
\end{table}

Table~(\ref{tab:MC dynamic}) reports the results of the Monte Carlo simulations. Compared to the static case in Table~(\ref{tab:MC static}), adding dynamics into the regression further deteriorates fixed effect estimators of strictly exogenous covariates, which are comparable with the standard deviations. On the other hand, indirect fixed effect estimators correct the bias significantly. Compared to HBC, the reduction of bias is comparable but the standard deviation is smaller, which is consistent with the theory: by construction HBC does not use the whole sample for bias correction and thus inflates the variance. 

\begin{table} [htbp!]
\centering
\begin{threeparttable}
\caption {\label{tab:MC dynamic} Simulation Results for Dynamic LFP } 
\begin{tabular}{rrrrrrrrrrrrr}
\hline\hline
  & \multicolumn{3}{c}{FE} &  & \multicolumn{3}{c}{IFE--10} & \multicolumn{3}{c}{HBC}   \\ 
\cline{2-4} \cline{6-8} \cline{10-12} 
            &  Bias   & Std Dev  &  Cvg & &  Bias  & Std Dev  &  Cvg & &  Bias  & Std Dev  &  Cvg\\    \hline
 lfp\_lagged & -53.59 & 5.84 & 0.50 & & 3.06 & 6.22 & 0.91 & & -6.43 & 7.32 & 0.92 \vspace{0.08cm}\\
 kids0\_2   &  33.45  & 13.64     & 0.62 & & -5.81  & 9.69 & 0.96  & & 7.62 & 17.27 & 0.97 \vspace{0.08cm}\\       
 kids3\_5   &  47.88  & 24.37    & 0.70 & & -8.53  & 18.65    & 0.96 & & 24.14 & 31.89 & 0.96 \vspace{0.08cm}\\         
 kids6\_17  &  53.38  & 73.44    & 0.91 & & -23.29 & 55.91    & 0.97 & & 33.74 & 98.06 & 0.97 \vspace{0.08cm}\\
 loghusinc  &  24.08  & 28.90    & 0.90 & & 5.29  & 44.90    & 0.98 & & 5.70 & 31.67 & 0.98 \vspace{0.08cm}\\
 age        &  29.49  & 19.34    & 0.84 & & 1.44   & 5.54     & 0.97 & & -1.46 & 33.73 & 0.97 \vspace{0.08cm}\\
 age2       &  29.07  & 26.91    & 0.86 & & -1.67  & 20.75    & 0.98 & & -1.04 & 36.54 & 0.97 \vspace{0.08cm}\\
 \hline
 \bottomrule
\end{tabular}
\begin{tablenotes}[flushleft]
\linespread{1}\footnotesize
\item\hspace*{-\fontdimen2\font}\textit{Notes:} FE denotes fixed effects estimates. IFE--10 denotes indirect fixed effect estiamtes with $H=10$. HBC denotes split--sample jackknife method. Simulations are conducted 1000 times, and all relative statistics are multiplied by 100. The nominal coverage is 95\%.
\end{tablenotes}
\end{threeparttable}
 \end{table}
\end{document}